\documentclass[11pt]{article}
\usepackage{graphicx} 
\usepackage{coding}
\begin{document}

\title{Coded Many-User Multiple Access  via \\  Approximate Message Passing}
\author{Xiaoqi Liu\thanks{X.\@ Liu and R.\@ Venkataramanan are with the Department of Engineering, University of Cambridge, UK. Author emails: \texttt{xl394@cam.ac.uk}, \texttt{kuanhsieh95@gmail.com}, \texttt{rv285@cam.ac.uk}.  This paper was presented in part at the 2024 IEEE International Symposium on Information Theory.} \;\;\;\;\; Kuan Hsieh \;\;\;\;\; Ramji Venkataramanan$^*$}
\maketitle
% \thanks{K.\@ Hsieh is with Luminance Technologies Ltd, UK. Email: \texttt{kuanhsieh95@gmail.com}

%

\begin{abstract}
We consider communication over the Gaussian multiple-access channel in the regime where the number of users grows
linearly with the codelength. In this regime,  schemes based on sparse superposition coding can achieve a near-optimal tradeoff between spectral efficiency and signal-to-noise ratio. However, these schemes are feasible only for small values of user payload. This paper investigates efficient schemes for larger user payloads, focusing on  coded CDMA schemes where each user's information is encoded via a linear code before being modulated with a signature sequence. We propose an efficient approximate message passing (AMP) decoder that can be tailored to the structure of the linear code, and provide an exact asymptotic characterization of its performance. 
Based on this result, we consider  a decoder that integrates AMP and belief propagation and characterize its tradeoff between spectral efficiency and signal-to-noise ratio, for a  given target error rate. Simulation results show that the decoder  achieves state-of-the-art performance at finite lengths, with a coded CDMA scheme defined using LDPC codes and a spatially coupled matrix of signature sequences.
\end{abstract}

\section{Introduction}
 We consider communication over an $L$-user Gaussian multiple access channel (GMAC), which has  output of the form 
\begin{equation}\label{eq:standard_gmac}
\by = \sum_{\ell=1}^L \bc_\ell+\bveps \, ,   
\end{equation}
over $n$  channel uses. Here
$\bc_\ell\in\reals^n$ is the codeword of the  $\ell$-th user and $\bveps\sim \normal_n(\0, \sigma^2\bI)$ is the channel noise.   Motivated by modern applications in  machine-type communications, a  number of recent works have studied the GMAC in the \emph{many-user} or \emph{many-access} setting, where the number of users $L$ grows with the block length $n$  \cite{chen2017capacity, polyanskiy2017perspective, zadik2019improved,RaviK22}.

In this paper,  we study the many-user regime where $L, n\to \infty$ with the user density $\mu: = L/n$ converging to a  constant. 
 Each user transmits a fixed number of bits $k$ (payload) under a constant energy-per-information-bit constraint $\|\bc_\ell\|^2_2/k\le E_b$. The spectral efficiency is the total user payload per channel use, denoted  by $\S=(Lk)/n = \mu k$. In this regime, a key question is to understand the tradeoff between user density (or spectral efficiency), the signal-to-noise ratio $E_b/N_0$, and  the probability of decoding error. Here $N_0=2 \sigma^2$ is the noise spectral density.   
A popular measure of decoding performance is the per-user probability of error ($\PUPE$), defined as 
\begin{equation}
    \PUPE := \frac{1}{L} \sum_{\ell=1}^{L} \, \prob( \bc_\ell \neq \hat{\bc}_\ell),
    \label{eq:PUPE_def}
\end{equation}
where $\hat{\bc}_\ell$ is the decoded codeword for user $\ell$.

Polyanskiy \cite{polyanskiy2017perspective} and Zadik et al.\@  \cite{zadik2019improved} obtained converse and achievability bounds on the minimum ${E_b}/{N_0}$ required to achieve  $\PUPE \leq \epsilon$ for a given $\epsilon>0$, when the user density $\mu$ and user payload $k$ are fixed. These bounds were extended to the multiple-access channels with Rayleigh fading in  \cite{kowshik2019isit,kowshikPoly21}.
The achievability bounds in these works are obtained using Gaussian random codebooks and joint maximum-likelihood decoding, which is computationally infeasible. In contrast, the focus of our work is computationally efficient schemes for which the  tradeoff above can be precisely characterized.

\paragraph*{Coding schemes based on sparse superposition coding}

Sparse superposition codes were introduced by Barron and Joseph \cite{joseph2012least, joseph2014fast} for the 
\emph{single-user} Gaussian channel, but also give a useful framework for efficient communication over the many-user GMAC \cite{hsieh2022near}.  We briefly review sparse superposition coding in the context of the GMAC.  The codeword for user $\ell \in [L]$ is constructed as $\bc_\ell = \bA_\ell \bx_\ell$, where $\bA_\ell \in \reals^{n\times B}$ is a random matrix and $\bx_\ell \in \reals^{B}$ is a message vector with exactly one nonzero entry. (The value of the nonzero entry is pre-specified, and chosen to satisfy the energy constraint.) Since there are $B$ choices for the location of the nonzero, each user transmits $\log_2 B$ bits in $n$ uses of the GMAC, and the spectral efficiency is $\S=\mu \log_2 B$. For this construction, the combined channel model can be written as
\begin{equation}
    \label{eq:linear_regression}
\by =  \sum_{\ell=1}^L \bA_\ell \bx_\ell  + \bveps = \bA\bx + \bveps,
\end{equation}
where the design matrix $\bA\in\mathbb{R}^{n\times L B}$ is the horizontal concatenation of matrices $\bA_1, \ldots, \bA_L$,
and the message vector $\bx\in\mathbb{R}^{LB}$ is the concatenation of vectors $\bx_1, \ldots, \bx_L$. 

The decoding problem is to recover the message vector $\bx$ from $(\by, \bA)$. An efficient Approximate Message Passing (AMP) decoder for this sparse superposition  GMAC scheme was analyzed in \cite{hsieh2022near}. AMP is a family of iterative algorithms that has its origins in relaxations of belief propagation \cite{Kab03,boutrosCaire2002, donoho2009message}. An attractive feature of AMP decoding is that it allows an exact asymptotic characterization of its error performance through a deterministic recursion called `state evolution'.  For both i.i.d.\@ and spatially coupled choices of the design $\bA$,  the asymptotic tradeoff achieved by the AMP decoder between the  spectral efficiency and $E_b/N_0$  was precisely characterized in \cite{hsieh2022near} (for a target $\PUPE$). In fact, \cite{hsieh2022near} analyzed a more general scheme  where each user's message vector $\bx_\ell$ in \eqref{eq:linear_regression} can be drawn from a general discrete prior on $\reals^B$; sparse superposition coding is an important special case of this scheme.  

Fig.\@ \ref{fig:bounds_k8} 
% \RV{1(a), figure similar to Fig. 2(b) of \cite{hsieh2022near}, no empirical performance plots} 
shows the tradeoff for the sparse superposition scheme with AMP decoding for user payload $k=8$ and  target $\PUPE=10^{-4}$. We observe that the performance of the scheme with a spatially coupled design (dotted black plot) is close to the converse bound (red plot), i.e., it is nearly optimal for $\S > 1.6$, and is uniformly better than the achievability bound from \cite{zadik2019improved} (solid black plot). However, sparse superposition coding  is challenging to scale to large user payloads, e.g., $k=240$ bits. To see this, we recall that each user has a message vector of size $B=2^k$, so the size of the design matrix, and hence the AMP decoding complexity, grows exponentially with $k$. One solution is to divide each user's payload of $k$ bits into smaller chunks of $\tilde{k}$ bits, which are transmitted sequentially using  $k/\tilde{k}$ blocks of transmission. Fig.\@ \ref{fig:bounds_k240}
% \RV{insert figure for similar numbers} 
illustrates the performance   for $k=240$ and $\tilde{k}=8$. We observe that there is now a significant gap between the performance of the sparse superposition scheme and the achievability and converse bounds from \cite{zadik2019improved}.
This motivates the question studied in this paper: for large payloads, can we construct efficient coding schemes whose asymptotic performance is closer to the bounds?

\begin{figure}
    \centering
    \subfloat[$k=8$]{\includegraphics[width=0.5\linewidth]{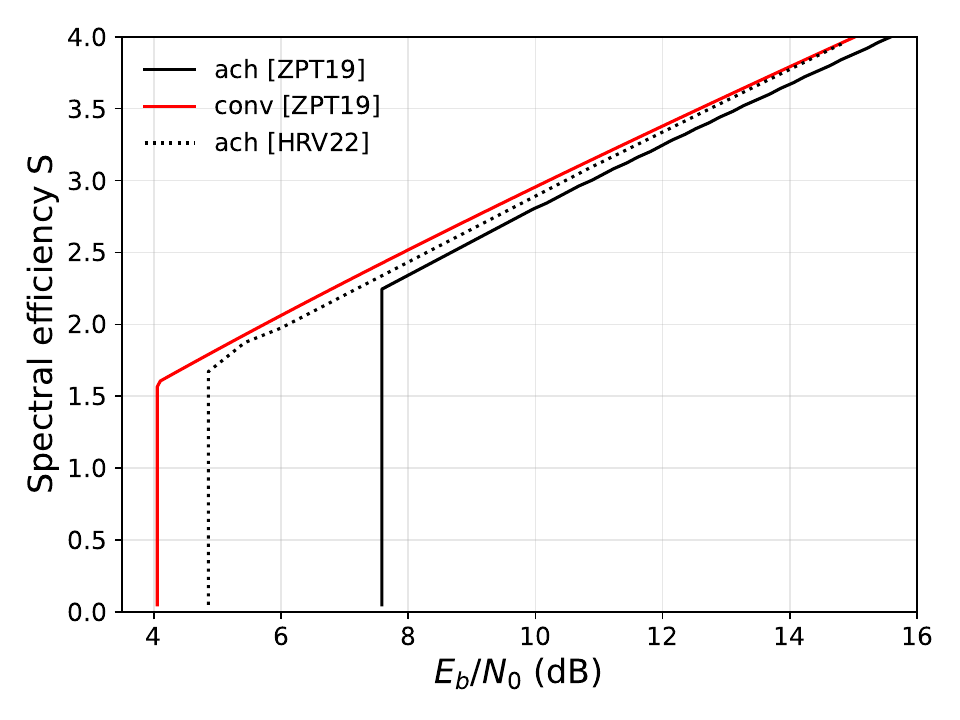}\label{fig:bounds_k8}}
    \subfloat[$k=240$]{\includegraphics[width=0.5\linewidth]{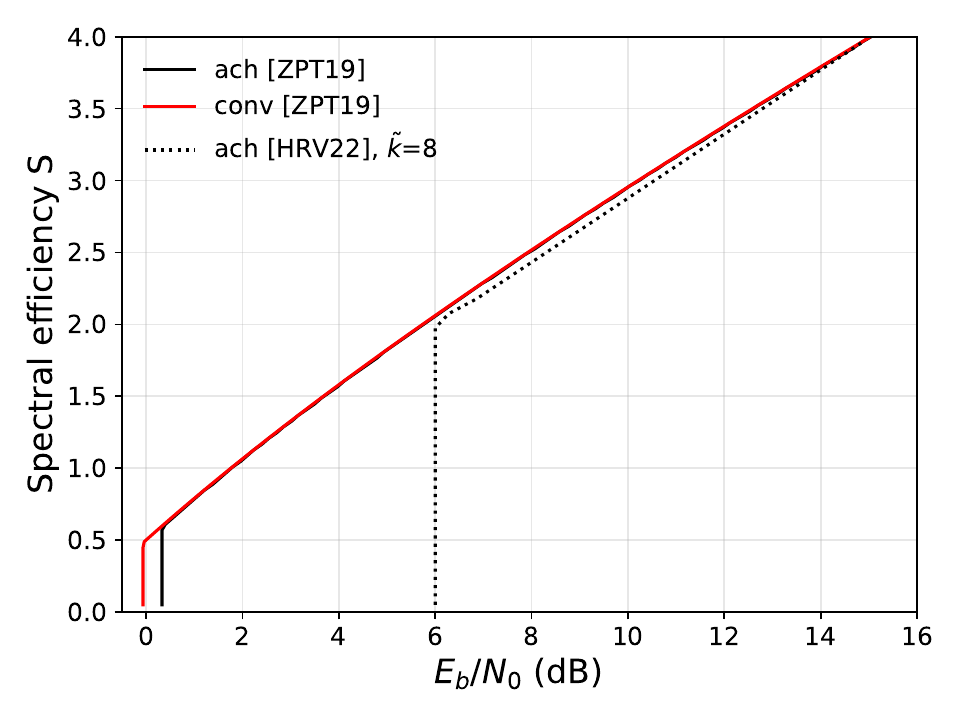} \label{fig:bounds_k240}}
    \caption{Asymptotic performance of spatially coupled sparse superposition scheme from \cite{hsieh2022near} (dotted black) 
    %and that of uncoded binary CDMA scheme (dashed black)
    compared with the asymptotic achievability and converse bounds from \cite{zadik2019improved} (solid black and red),  for user payload  $k=8$ or $k=240$, and target $\PUPE=10^{-4}$. }
% main_AB_journal_k8_k240.py
\end{figure}

\paragraph{Coding schemes based on CDMA}
Our starting point is binary code-division multiple access (CDMA), where each user $\ell \in [L]$ transmits one bit of information  by modulating a signature sequence (or spreading sequence) $\ba_\ell \in \reals^n$, i.e., the codeword $\bc_\ell = \ba_\ell x_\ell$ where $x_\ell \in \{ \pm \sqrt{E} \}$. The
decoding problem is to recover the vector of information symbols $\bx= [x_1, \ldots, x_L]^{\top}$ from the channel output vector
\begin{equation}
    \by = \sum_{\ell=1}^L \ba_\ell x_\ell  + \bveps =
    \bA \bx  +  \bveps \, ,
    \label{eq:binCDMA}
\end{equation}
where $\bA= [\ba_1, \ldots, \ba_L] \in \reals^{n\times L}$ is the matrix of signature sequences. We have used the same notation as 
in \eqref{eq:linear_regression} to highlight that binary CDMA is similar to the model in \eqref{eq:linear_regression},  with $B=1$. 
% \SL{\eqref{eq:linear_regression} uses no modulation, so maybe we should say  ``a special case ... with  $B=1$ and $\{x_\ell\in \pm\sqrt{E}\}$''?}

The optimal spectral efficiency of CDMA in the large system limit (with random signature sequences)  has been studied in a number of works, e.g.,
\cite{verduShamai99,shamaiVerdu01,CaireGuemRomyVerdu2004,tanaka2002cdma,guo2005randomly}. Assuming the signature sequences are i.i.d.\@ sub-Gaussian, the best known technique for efficiently decoding $\bx$ from $\by$ in \eqref{eq:binCDMA} is AMP \cite{donoho2009message, bayati2011dynamics}.  If each  user wishes to transmit $k >1$ bits in $n$ channel uses, the binary CDMA scheme requires $k$ blocks of transmission, with each block (and each signature sequence) having length $n/k$.
For large user payloads $k$, the binary CDMA  with AMP decoding performs poorly, with a tradeoff curve that is significantly worse than the sparse superposition scheme.

\subsection{Main contributions} The binary CDMA scheme described above transmits \emph{uncoded} user information.
     In this paper,  we show how a significantly better performance tradeoff can be achieved for large user payloads  (compared to the sparse superposition scheme and uncoded CDMA) using a coded CDMA scheme in which each user's information sequence is first encoded using a linear code before being multiplied with the signature sequence.
Our scheme uses an efficient AMP decoder that is tailored to the decoder of the linear code. Among potential linear codes, Low-Density Parity Check (LDPC) codes are a particularly effective choice. They are known to be capacity-achieving for binary input symmetric channels with efficient Belief Propagation (BP) decoding \cite{richardson_urbanke_2008}. In the coded CDMA setting, the soft-information available from the BP decoder can be leveraged within the AMP framework to iteratively refine the receiver's estimate of the user information.% 

% \KH{In this paper, we show that by using a coded CDMA scheme in which each user's information sequence is first encoded using a linear code before being multiplied with the signature sequence, and an AMP decoder that is tailored to the linear code's decoder, a significantly better performance tradeoff can be achieved at large user payloads compared to the sparse superposition scheme and uncoded CDMA scheme.
% Among potential linear codes, Low-Density Parity Check (LDPC) codes are a particularly effective choice. They are known to be capacity-achieving for binary input symmetric channels, and the soft-information available from Belief Propagation (BP) decoding can be leveraged within the AMP framework to iteratively refine the decoder's estimation of the user information.}

Coded CDMA schemes for multiple access (with different decoders) have been studied in earlier works \cite{reed1998iterative, turbo1999wang, caire2002iterative, caire2004iterative, ldpc2005wang}, which we discuss in the next subsection. 
The novel contribution of our work is a flexible AMP decoder that can be tailored to the structure of the linear code, and  an exact asymptotic characterization of its error performance in the many-user regime (Theorems \ref{thm:UER_BER} and \ref{thm:UER_BER_sc}).  Specifically, we show how a decoder for the underlying code, such as a maximum-likelihood or a belief propagation (BP) decoder, can be incorporated within the AMP algorithm with rigorous asymptotic guarantees (Corollary \ref{corr:SEvalidity}). 
Simulation results validate the theory and demonstrate the benefits of the concatenated scheme at finite lengths.

We will use a \emph{spatially coupled} construction for the design matrix $\bA$. Spatial coupling was originally proposed \cite{felstrom1999time} as a way to improve the decoding threshold of LDPC codes. Spatially coupled LDPC codes with belief propagation have  since been shown to be capacity-achieving  decoding for a large class of binary-input channels \cite{kudekar2013spatially}. Spatially coupled sparse superposition codes with AMP decoding have also been shown to be capacity-achieving, for both AWGN \cite{barbier2017approximate,rush2021capacity} and a broader class of channels \cite{barbier2019universal}. Moreover, spatially coupled designs (with estimation via AMP)  achieve Bayes-optimal error for both linear \cite{krzakala2012statistical, donoho2013information,takeuchi2015performance} and   generalized linear models \cite{cobo2023bayes}. To keep the exposition simple, we first present the AMP decoder and analysis for the i.i.d.\@ Gaussian design in Sections \ref{sec:iid_AMP} and \ref{sec:iid_denoisers} before generalizing the results to spatially coupled designs in Section \ref{sec:SC_AMP}.

We emphasize that our setting is distinct from unsourced random access over the GMAC \cite{polyanskiy2017perspective, Amalladinne2020Coded,fengler2021SPARCS,amalladinne2022unsourced}, 
where all the users share the same codebook and only a subset of them are active. In our case, each user has a distinct signature sequence and all of them are active. While the latter is particularly relevant in designing grant-free communication systems,  coding schemes  for this setting often rely on   dividing a common codebook into sections for different users. 
Extending the ideas in this paper to unsourced random access is a promising direction for future work.

 \subsection{Related work} 
 %\SL{added two paragraphs below}
\paragraph*{Coded CDMA for multiple access channels} 

Efficient decoders for coded CDMA schemes  have been studied in  a number of works \cite{reed1998iterative, turbo1999wang, caire2002iterative, caire2004iterative, ldpc2005wang}. These decoders are typically based  on iterative soft interference cancellation, and have two components: a   multiuser detector,  and a single-user channel decoder for each user. The idea is to iteratively exchange soft-information between the two components to  approximate the posterior probabilities on the information symbols.  In the spirit of these decoders, one could use an AMP algorithm for multiuser detection followed by  single-user channel decoding (e.g., via BP).  Our results show that we  can significantly improve on this approach by integrating the channel decoder within the AMP algorithm (see Fig. \ref{fig:bp_v_marginal}). 

Optimal power allocation across users for iterative schemes based on multiuser detection and single-user decoding was studied in \cite{caire2004iterative}. We expect that the spatial coupling we use has a similar effect to optimal power allocation. This is based on a similar phenomenon observed for power-allocated vs.\@ spatially coupled  sparse superposition codes for single-user AWGN channels \cite{venkataramanan19monograph}.   We note that our scheme can be viewed as an instance of non-orthogonal multiple access (NOMA) \cite{Ding2017,Yue2018}, a class of schemes based on superposition coding that has been studied widely in recent years.

\paragraph*{Approximate Message Passing}AMP algorithms were first proposed for compressed sensing \cite{donoho2009message, bayati2011dynamics} and its variants  \cite{Ziniel2013efficient}.  These algorithms    have since been applied to a range of problems including estimation in generalized linear models and low-rank matrix estimation.  We refer the interested reader to \cite{feng2022unifying} for a survey. In the context of communication over AWGN channels, AMP has been used as a decoder for sparse superposition codes (SPARCs) \cite{barbier2017approximate, rush2017capacity,venkataramanan19monograph} and for compressed coding \cite{Liang20Compressed}. For the GMAC, in addition to  efficient schemes for smaller user payloads \cite{hsieh2022near}, spatially coupled SPARCs  have been used to obtain improved  achievability bounds \cite{kowshik2022improved, liu2024RA}.   SPARC-based concatenated schemes with AMP decoding have been proposed for both single-user AWGN channels \cite{greig2018techniques, cao21sparclist, ebert2025sparse} and  the GMAC where users have different codebooks \cite{liu2024RA} or share the same codebook \cite{fengler2021SPARCS, amalladinne2022unsourced}.

 In most of these concatenated schemes \cite{greig2018techniques, cao21sparclist, fengler2021SPARCS, Liang20Compressed} the AMP decoder for SPARCs does not explicitly use the structure of the outer code (which is decoded separately).  Two key exceptions are the SPARC-LDPC concatenated schemes  in \cite{amalladinne2022unsourced,ebert2025sparse}, which use an AMP decoding algorithm with an  integrated BP denoiser. Drawing inspiration from these works,  in Section \ref{subsec:BPdenoiser} we propose an AMP decoder with a BP denoiser for our concatenated scheme. Our scheme and its decoder differ from those in \cite{amalladinne2022unsourced,ebert2025sparse} in a few important ways: i) we do not use the SPARC message structure, and ii)  we treat each user's codeword as a row  of a signal matrix and devise an AMP algorithm with matrix iterates, a notable deviation from prior schemes where AMP operates on vectors. 

 The SPARC-LDPC scheme was recently extended to multiple access channels in \cite{ebert2024multiusersrldpccodes}. Since the scheme is based on sparse superposition codes, for the reasons discussed on p.\pageref{eq:PUPE_def}, the AMP decoder in \cite{ebert2024multiusersrldpccodes} can only handle a small number of users if the per-user payload is large (or a small payload with a large number of users). In contrast, our focus is to design an efficient scheme for a large number of users,  each with a  large fixed payload, e.g., $2000$ users with a payload of $300$ bits each.
 
 %\SL{Could mention \cite{ebert2024multiusersrldpccodes} here as well} \RV{done.} 
 %\SL{would it be better to say ``a large number of users, each with a fixed large payload.''}

\subsection{Notation} We write $[L]$ for the set $\{1, \ldots, L \}$.
 We use bold uppercase letters  for matrices, bold lowercase  for vectors, and  plain font for scalars.   We write $\ba_\ell$ for the $\ell$-th row or column of  $\bA$ depending on the context,  and $a_{\ell, i}$ for its $i$-th component.  A function $f:\reals^d\to \reals^d$ returns a column vector when applied to a column vector, and likewise for row vectors. 

\section{Concatenated coding scheme} 

The $k$-bit message of user $\ell$, denoted by $\bu_\ell\in \{0,1\}^{k}$, is mapped to a GMAC codeword $\bc_\ell \in \reals^{n}$ in two steps. First,  a rate ${k}/{d}$ linear code with generator matrix $\bG\in\{0,1\}^{ d \times  k}$ is used to produce a $d$-bit  binary codeword $\bG\bu_\ell\in\{0,1\}^{ d}$. Each $0$ code bit is then mapped to $\sqrt{E}$ and each $1$ bit code bit to $-\sqrt{E}$ to produce $\bx_\ell \in \{ \pm \sqrt{E} \}^d$. The magnitude $\sqrt{E}$ of each BPSK symbol will be  specified later in terms of the energy per bit constraint $E_b$. In the second step of encoding, for each user $\ell$, we take the outer-product of  $\bx_\ell$ with a  signature sequence  $\ba_\ell\in\reals^{\tn}$,  where $\tn:=n/d.$  This yields a matrix $\bC_\ell=\ba_\ell \bx_\ell^\top\in\reals^{\tn\times d}$.  The   final length-$n$ codeword  transmitted by user $\ell$ is simply $\bc_\ell=\vectorize(\bC_\ell) \in \reals^n$.

 \begin{figure}[t]
 \centering
 \includegraphics[width=0.7\linewidth]
 {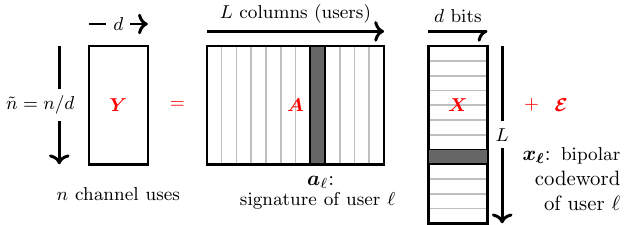}
\caption{Concatenated coding scheme for GMAC}
\label{fig:Y=AX+W}
\vspace{-10pt}
 \end{figure} 
 
Let
$\bX=
\begin{bmatrix}
    \bx_1, \dots, \bx_L 
\end{bmatrix}^\top \in\{\pm \sqrt{E}\}^{L\times d} $
be the signal matrix whose $\ell$-th row $\bx_\ell$ is the bipolar codeword of user $\ell$.
%which stacks together the users' codewords obtained via the outer LDPC code.
% 
Let $\bA = \begin{bmatrix}
    \ba_1, \dots, \ba_L
\end{bmatrix} \in\reals^{\tn\times L}$
be
the design matrix whose columns are the signature sequences.  Then the channel output in \eqref{eq:standard_gmac} can be rewritten into matrix form: 
%\vspace{-0.2cm}
\begin{equation}
\bY 
=\sum_{\ell=1}^L\ba_\ell \bx_\ell^\top +\bEps
= \bA\bX +\bEps\in\reals^{\tn\times d}. \label{eq:Y=AX+W}
\end{equation}
See Fig.\@ \ref{fig:Y=AX+W} for an illustration.

\subsection{Assumptions} \label{subsec:assumptions}

We consider random  signature sequences with independent entries and expected squared $\ell_2$-norm of  one, i.e., $\E\|\ba_\ell\|_2^2=1$. For example, with i.i.d.\@ Gaussian signature sequences, the entries of the design  matrix are  $A_{i\ell} \stackrel{\text{i.i.d.}}{\sim} \normal(0, 1/\tn)$ for $i \in [\tn], \,  \ell \in [L]$.
Due to concentration of measure, this implies that $\|\ba_\ell\|_2^2 \to 1$ as $n\to\infty$. Hence, the energy-per-bit constraint can be   enforced asymptotically by choosing $E$ so that $\|\bc_\ell\|_2^2/k = \|\ba_\ell\bx_\ell^\top\|_{\text{F}}^2/k \to Ed/k\le E_b$ as $n\to\infty$.  We make the natural assumption that the information bits $\bu_\ell \in  \{ 0,1 \}^k$ are uniformly random, for each user $\ell \in [L]$. We also assume that the noise variance $\sigma^2$ in \eqref{eq:standard_gmac} is known, a mild assumption since $\sigma^2$ can be consistently estimated as:
\begin{align}
    \widehat{\sigma}^2 &=\frac{\| \bY \|^2_{\text{F}}}{n} -\frac{1}{n}\sum_{i\in[\tn]}\sum_{j\in[d]}\E\left[\bigg(\sum_{\ell\in[L]} A_{i\ell} X_{\ell j}\bigg)^2\right] \nonumber\\
    & = \frac{\| \bY \|^2_{\text{F}}}{n} - \frac{1}{n} \sum_{i\in[\tn]}\sum_{j\in[d]}  \sum_{\ell\in[L]}\E\left[ A_{i\ell}^2\right]\cdot E\nonumber\\
    &  = \frac{\| \bY \|^2_{\text{F}}}{n} -  d\mu  E, \nonumber
\end{align}
where the last equality uses $ \sum_{i\in[\tn]} \E\left[ A_{i\ell}^2\right] = \E\|\ba_\ell\|_2^2=1$, for $\ell \in [L]$.

We consider the  asymptotic limit where $L/n \to \mu$ as $n,L \to \infty$, for a user density $\mu >0$ of constant order.  We emphasize that $d$ is fixed and does not scale as $n,L \to \infty$. Therefore $\tn/L = ({n}/{d})/L  \to {1}/{(d\mu)}$ is also of constant order.

\section{AMP decoder for i.i.d.\@ design}\label{sec:iid_AMP}
Consider the  i.i.d.\@ Gaussian design matrix $\bA$ with entries $\stackrel{\text{i.i.d.}}{\sim}\normal(0, 1/\tn)$. 
The decoding task is to recover the signal matrix $\bX$  from the channel observation $\bY$ in \eqref{eq:Y=AX+W}, given the design matrix $\bA$ and the channel noise variance $\sigma^2$ (or its estimate).  A good decoder must take advantage of the prior distribution on $\bX$: recall that each row of $\bX$ is an independent codeword taking values in $\{ \pm \sqrt{E} \}^d$, defined via  the underlying rate $k/d$ linear code.  The prior distribution of each row of $\bX$ induced by the linear code is denoted by $P_{\bar{\bx}}$. Note that $P_{\bar{\bx}}$ assigns equal probability to  $2^k$  vectors in $\{ \pm \sqrt{E} \}^d$.

The AMP decoder recursively produces estimates $\bX^t\in \reals^{L\times d}$ of $\bX$ for iteration $t\ge 0$. This is done via a sequence of denoising functions $\eta_t$ that can be tailored to the prior $P_{\bar{\bx}}$.  Starting from an initializer $\bX^0=\bzero_{L\times d}$, for $t \ge 0$ the AMP decoder computes: 
\begin{align}
&\bZ^t = \bY -\bA\bX^t +\frac{1}{\tn}\bZ^{t-1} 
\left[\sum_{\ell=1}^L \eta_{t-1}'\left(\bs^{t-1}_\ell\right)\right]^\top ,\label{eq:alg_Zt}\\
&\bX^{t+1} =\eta_t\left(\bS^t\right),\quad \bS^t = \bX^t+\bA^\top \bZ^t,\label{eq:alg_Xt}
\end{align}
where $\eta_t: \reals^d\to \reals^d$ applies row-wise to matrix inputs, and $\eta_t'(\bs) = \frac{\de \eta_t(\bs)}{\de \bs}\in\reals^{d\times d}$ is the derivative (Jacobian) of $\eta_t$.  Quantities with negative indices are set to all-zero matrices. 
When $d=1$, \eqref{eq:alg_Zt}--\eqref{eq:alg_Xt}  reduces to the  classical AMP algorithm \cite{bayati2011dynamics} for estimating a vector signal in a linear model.

\paragraph{State Evolution (SE)}  As $n,L \to \infty$ (with $\frac{L}{n} \to \mu$), the memory term $ \frac{1}{\tn}\bZ^{t-1} 
\left[\sum_{\ell=1}^L \eta_{t-1}'\left(\bs^{t-1}_\ell\right)\right]^\top$   term  in \eqref{eq:alg_Zt}
ensures that the row-wise empirical distribution of $\bZ^t \in \reals^{\tn \times d}$ converges to a Gaussian $\normal_d(\bzero, \bSigma^t)$ for $t \ge 1$. Furthermore, the row-wise empirical distribution of $(\bS^t - \bX) \in \reals^{L \times d}$ also converges to the same Gaussian $\normal_d(\bzero, \bSigma^t)$. The covariance matrix $\bSigma^t \in \reals^{d \times d}$ is iteratively  defined via the following state evolution (SE) recursion, for $t \ge 0$:
\begin{align}
&\bSigma^{t+1}  = \sigma^2\bI_d \, +  \, d\mu  \E\left\lbrace [ \eta_{t}(\bar{\bx} +\bg^{t})-\bar{\bx}] [\eta_{t}(\bar{\bx} +\bg^{t})-\bar{\bx}]^\top
\right\rbrace. \label{eq:SE_Sigma_k+1}
\end{align}
Here   $\bg^t\sim \normal_d(\bzero, \bSigma^t)$ is independent of $\bar{\bx} \sim P_{\bar{\bx}}$, and $\bI_d$ is the $d \times d$ identity matrix. The expectation in \eqref{eq:SE_Sigma_k+1} is with respect to $\bar{\bx}$ and $\bg^t$,  and the iteration is initialized with $\bSigma^0 = (\sigma^2 +d\mu E)\bI_d$. We shall refer to \eqref{eq:SE_Sigma_k+1} as `iid-SE' later on because it is associated with the i.i.d.\@ Gaussian design.

The  convergence of the row-wise empirical distribution  of $\bS^t$ to  the law of $\bar{\bx} + \bg^{t} $ follows by applying standard results in AMP theory \cite{javanmard2013state, feng2022unifying}.  This distributional characterization of $\bS^t$ crucially informs the  choice of the denoiser $\eta_t$. Specifically, for each row $\ell \in [L]$, the role of the denoiser  $\eta_t$ is to estimate the codeword $\bx_\ell$ from an observation in zero-mean Gaussian noise with covariance matrix $\bSigma^t$.
In the next section, we  discuss the Bayes-optimal denoiser and two other  sub-optimal but computationally efficient  denoisers.

First, we provide a performance characterization of the AMP decoder with a generic Lipschitz-continuous denoiser in Theorem \ref{thm:UER_BER}. Decoding performance after  $t$ iterations of AMP decoding  can be measured via either the user-error rate $\UER =\frac{1}{L}\sum_{\ell=1}^L\ind\{\hat{\bx}^{t+1}_\ell \neq \bx_\ell\}$,
or the bit-error rate $ \BER=\frac{1}{Ld}\sum_{\ell=1}^L\sum_{j=1}^d \ind\{ \hat{x}^{t+1}_{\ell, j}\neq x_{\ell, j}\}$. Here  $\hat{\bx}^{t+1}_\ell = h_t(\bs^{t}_\ell)
$ is a hard-decision estimate of the codeword $\bx_\ell$, produced using a suitable function $h_t$, and $\hat{x}^{t+1}_{\ell, j}$  is the $j$th entry of $\hat{\bx}^{t+1}_\ell$.  For example, $h_t$ may quantize each entry of  $\bx^{t+1}_{\ell} = \eta_t(\bs^{t}_\ell)$ to a value in $\{\pm \sqrt{E}\}$. 
We note that the $\PUPE$ defined in \eqref{eq:PUPE_def} is the expected value of the $\UER$.

\begin{thm}[Asymptotic $\UER$ and $\BER$ with i.i.d.\@ design]\label{thm:UER_BER} 
Consider the concatenated scheme with an i.i.d.\@ Gaussian design matrix, with the assumptions in Section \ref{subsec:assumptions}, and the AMP decoding algorithm in \eqref{eq:alg_Zt}--\eqref{eq:alg_Xt} with  Lipschitz continuous denoisers $\eta_t:\reals^d \to \reals^d$, for $t \ge 1$. Let $\hat{\bx}^{t+1}_\ell = h_t(\bs^t_\ell)$ be the hard-decision estimate of $\bx_\ell$  in iteration $t$.
The asymptotic $\UER$ and $\BER$ in iteration $t$    satisfy the following almost surely, for $t \ge 0$:
\begin{align}\label{eq:UER_thm}
     \lim_{L\to \infty} \UER & := 
     \lim_{L\to \infty} \frac{1}{L}\sum_{\ell=1}^L\ind\left\lbrace\hat{\bx}^{t+1}_\ell \neq \bx_\ell\right\rbrace  
      = \prob \left( h_{t}\left(\bar{\bx}+ \bg^{t}\right) \neq \bar{\bx} \right),  \\
    \label{eq:BER_thm} 
   \lim_{L\to \infty} \BER & := 
      \lim_{L\to \infty}\frac{1}{Ld}\sum_{\ell=1}^L\sum_{j=1}^d \ind\left\lbrace \hat{x}^{t+1}_{\ell, j}\neq x_{\ell, j}\right\rbrace   = \frac{1}{d}\sum_{j=1}^d\prob \left( \left[h_{t}\left(\bar{x}+g^{t}\right)\right]_j\neq \bar{x}_{j}\right).
\end{align}
Here $\bar{\bx} \sim P_{\bar{\bx}}$  and $\bg^t\sim \normal_d(\bzero, \bSigma^t)$ are independent, with $\bSigma^t$ defined by the state evolution recursion in \eqref{eq:SE_Sigma_k+1}. The limit is taken as $n,L \to \infty$ with $L/n\to \mu$.
\end{thm}
\begin{proof}
The proof  is given in Section \ref{sec:proof_iid}.
\end{proof}
Theorem \ref{thm:UER_BER} states that in each iteration $t$, the empirical  $\UER$ and $\BER$ of the AMP decoder with a Lipschitz denoiser asymptotically converge to the deterministic quantities on the RHS of \eqref{eq:UER_thm} and \eqref{eq:BER_thm}, which involve the $d$-dimensional SE random vectors $\bbx$ and $\bg^t$.

\section{Choice of AMP denoiser $\eta_t$}\label{sec:iid_denoisers}

\subsection{Bayes-optimal denoiser}
Since the row-wise distribution of  $\bS^t$ converges to the law of $\bar{\bx}+ \bg^{t}$, the Bayes-optimal or minimum mean squared error (MMSE)  denoiser $\eta_t^{\Bayes}$ estimates each row $\bX$ as the following conditional expectation. For $\ell \in [L]$,
\begin{align}
    &\bx^{t+1}_\ell =\eta_t^{\Bayes}(\bs^t_\ell)  =  \E\left[\bar{\bx} \mid \bar{\bx}+ \bg^t= \bs^t_\ell\right] \nonumber\\
    & = \sum_{\bx' \in\mc{X}}\bx'\cdot \frac{\exp\left(-\frac{1}{2} (\bx' - 2\bs_\ell^t)^\top (\bSigma^t)^{-1} \bx'\right)}{\sum_{\tilde{\bx}' \in \mc{X}} \exp\left(-\frac{1}{2} (\tilde{\bx}' - 2\bs_\ell^t)^\top (\bSigma^t)^{-1} \tilde{\bx}'\right)}\label{eq:Bayes_denoiser}
\end{align}
where $\mc{X} \subset \{ \pm \sqrt{E} \}^d$ is the set of $2^k$ codewords.
%:= \{\bx = \sqrt{E}\cdot (-1)^{\bG\bu}\text{ for }\bu\in \{0,1\}^k\}$, enforcing the parity consistency condition on $\bar{\bx}$. 
% 
Since $\abs{\mc{X}} = 2^k$, 
the  cost of applying  $\eta_t^{\Bayes}$ is $O(2^kd^3)$ which grows exponentially in  $k$.
In each iteration $t$, the decoder can produce a hard-decision maximum a posteriori (MAP) estimate $\hat{\bx}_\ell^t$ from $\bs_\ell^t$ via:
\begin{align}
    \hat{\bx}_\ell^{t+1} = h_t(\bs_\ell^t) = \argmax_{\bx'\in \mc{X}} \prob\left(\bar{\bx} = \bx' \mid \bar{\bx} + \bg^t =\bs_\ell^t\right).\label{eq:MAP_est}
\end{align}

In Section \ref{sec:iid_numerical} (Fig.\@ \ref{fig:hamming})  we present  numerical results illustrating the performance of AMP with denoiser $\eta_t^{\Bayes}$ for a Hamming code with $d=7$ and $k=4$.
In practical scenarios where $d$ is of the order of several hundreds or thousands, applying $\eta_t^{\Bayes}$ is not feasible, motivating the use of sub-optimal denoisers with lower  computational cost.

\subsection{Marginal-MMSE denoiser}
 A computationally efficient alternative to the Bayes-optimal denoiser is the marginal-MMSE denoiser \cite{greig2018techniques, fengler2021SPARCS} which
acts   entry-wise on  $\bs_\ell^t$ and 
returns the entry-wise conditional expectation:
\begin{align}
    \bx_{\ell}^{t+1} = \eta_{t}^{\marginal}(\bs_\ell^t) & =
    \begin{bmatrix}
    \E[\bar{x}_1 \mid \bar{x}_1 + g^t_1 = s_{\ell,1}^t]\\
    \vdots\\
        \E[\bar{x}_d \mid \bar{x}_d + g^t_d = s_{\ell,d}^t] 
    \end{bmatrix},  \qquad \text{where}\nonumber\\
    \E[\bar{x}_j \mid  \bar{x}_j + g^t_j = s_{\ell,j}^t]& \stackrel{\text{(i)}}{=} \sqrt{E} \tanh 
    \left( \sqrt{E}s_{\ell,j}^t/\Sigma_{j,j}^t\right)\qquad \text{for }j\in[d]. \label{eq:marginal_MMSE_exp}
\end{align}
The equality (i) follows from $g^t_j \sim \normal(0, \Sigma^t_{j,j})$ and $p(\bar{x}_j=\sqrt{E}) =p(\bar{x}_j=-\sqrt{E})  =\frac{1}{2} $ due to the linearity of the outer  code.  A hard decision estimate $\hat{
\bx}_{\ell}^{t+1}$ can be obtained by quantizing each entry of $\bx_{\ell}^{t+1}$ to $\{ \pm \sqrt{E} \}$.

This marginal denoiser has an $O(d)$ computational cost which is linear in $d$, but it  ignores the parity structure of $\bar{\bx}$, which is useful prior knowledge that can help reconstruction. One way to address this is by using the output of the AMP decoder as input to a channel decoder for the outer code, as in \cite{fengler2021SPARCS, cao21sparclist}. 
In the next subsection, we show how to improve on this approach. Considering an outer LDPC code, we use an AMP denoiser that fully integrates BP decoding. 

 \subsection{Belief Propagation (BP) denoiser} \label{subsec:BPdenoiser}
 Assume that the binary code used to define the  concatenated scheme is a low-density parity check (LDPC) code. 
An LDPC code is a linear code characterized by a sparse parity-check matrix, where each parity-check equation involves only a small number of code  bits relative to the codelength. LDPC codes are typically decoded using Belief Propagation (BP), an iterative message passing algorithm that operates on the sparse factor graph  corresponding to the parity check matrix.  The factor graph is a bipartite graph with two sets of nodes: variable nodes (representing the code bits) and check nodes (representing the parity-check equations). The BP decoder iteratively passes messages between these two sets of nodes to compute updated probabilistic estimates on the value of each code bit.  For a detailed discussion of LDPC codes and BP decoding, see \cite{richardson_urbanke_2008}.
% \KH{An LDPC code is a linear code where the parity-check matrix that defines the code is of ``low degree'', i.e., the number of codeword bits contained in each parity-check equation is small compared to the codeword length.} 
% \SL{An LDPC code is a linear code characterized by a sparse parity-check matrix, where each parity-check equation involves only a small number of code  bits relative to the codelength. Belief Propagation (BP) operates on the corresponding sparse factor graph through iterative message passing between variable nodes (representing code bits) and check nodes (representing parity-check equations), computing probabilistic estimates of the transmitted codeword.}
% \KH{}

We propose a BP denoiser $\eta_t^{\BP}$ for the AMP decoder which exploits the parity structure of the  LDPC code in each AMP iteration by performing a few rounds of BP  on the associated factor graph. Like the other denoisers above, $\eta_{t}^{\BP}$ acts row-wise on the effective observation $\bS^t \in \reals^{L \times d}$. For $\ell \in [L]$, it  produces the updated AMP estimate $\bx_{\ell}^{t+1}$ from $\bs_\ell^t$  as follows, using $\mc{R}$ rounds of BP. 

1)  For each variable node $j\in [d]$ and check node $i\in [d-k]$, initialize variable-to-check messages (in log-likelihood ratio format) as:
\begin{align}
    L_{j\to i}^{(0)}=\ln\left[\frac{p(s_{\ell,j}^t  \mid x_{\ell,j} =+\sqrt{E})}{p(s_{\ell,j}^t  \mid x_{\ell,j} =-\sqrt{E})}\right] 
    = \frac{2\sqrt{E}s_{\ell,j}^t}{\Sigma^t_{j,j}}  =: L(s_{\ell,j}^t).
     \label{eq:initial_LLR}
\end{align}
This initialization follows the distributional assumption $s_{\ell,j}^t \stackrel{\text{d}}{=}\bar{x}_j + g^t_j$, where $g^t_j \sim \normal(0, \Sigma^t_{j,j})$ and $p(\bar{x}_j=\sqrt{E}) =p(\bar{x}_j=-\sqrt{E})  =\frac{1}{2} $. Note that, similar to the marginal-MMSE denoiser in \eqref{eq:marginal_MMSE_exp}, only the diagonal entries of the covariance matrix $\bSigma^t$, $\Sigma^t_{j,j}$ for $j\in[d]$,  are used in \eqref{eq:initial_LLR}. 

2) Let $N(i)$ denote the set of neighbouring nodes of node $i.$ For round $1 \le r \le \mc{R}$, compute the check-to-variable  and variable-to-check messages, denoted by $L_{ i\to j}^{(r)}$ and $L_{j\to i}^{(r)}$, as:
\begin{align}
    L_{ i\to j}^{(r)} & = 2\tanh^{-1}\left[\prod_{j'\in N(i)\setminus j} \tanh \left(\frac{1}{2}L_{j'\to i}^{(r -1)} \right)\right]\,, \label{eq:check_to_var} \\
    L_{j\to i}^{(r)} & = L(s_{\ell,j}^t) + \sum_{i'\in N(j)\setminus i}L_{i'\to j}^{(r)}\,.\label{eq:var_to_check}
\end{align}

3) Terminate BP after $\mc{R}$ rounds by computing the final log-likelihood ratio for each variable node $j \in [d]$: 
 \begin{align}
 L_{j}^{( \mc{R} )}   = L(s_{\ell,j}^t) + 
 \sum_{i\in N(j)}  L_{i\to j}^{(\mc{R})}.
\label{eq:LLR_last_step}  
\end{align}
Equations \eqref{eq:check_to_var}--\eqref{eq:LLR_last_step}  are the standard BP updates for an LDPC code \cite{richardson_urbanke_2008}. 

4) Compute the updated AMP estimate  $ \bx_{\ell}^{t+1} = \eta_t^{\BP}(\bs_\ell^t)$, where for $j\in[d]$, 
\begin{align}
    \left[\eta_t^{\BP}(\bs_\ell^t)\right]_j &= \sqrt{E}\ \frac{\exp\big(L_{j}^{(\mc{R})}\big)}{1+\exp\big(L_{j}^{( \mc{R} )} \big)} -\sqrt{E}\ \frac{1}{1+\exp\big(L_{j}^{( \mc{R} )} \big)}
    = \sqrt{E}\, 
    \tanh\big(L_{j}^{(\mc{R})}/2 \big).
    \label{eq:etat_BP_exp}
\end{align}
The RHS above is obtained by converting the final log-likelihood ratio \eqref{eq:LLR_last_step} to a conditional expectation, recalling that $x_{\ell, j}$ takes values in $\{ \pm \sqrt{E} \}$. Following the standard interpretation of BP as approximating the bit-wise marginal posterior probabilities \cite{richardson_urbanke_2008}, the expression  in \eqref{eq:etat_BP_exp} can be viewed as an approximation to $$\E[\bar{x}_j \mid  \bar{x}_j + g_j^t = s_{\ell, j}^t \, , \,  \text{parities specified by $N(j)$ 
 are  satisfied}].$$
We highlight the contrast between the  conditional expectation above and the one in \eqref{eq:marginal_MMSE_exp}, which does not use the parity-check constraints. As with the marginal-MMSE denoiser, a hard-decision estimate $\hat{
\bx}_{\ell}^{t+1}$ can be obtained by quantizing each entry of $\bx_{\ell}^{t+1}$ to $\{ \pm \sqrt{E} \}$. The computational cost of $\eta^{\BP}_t$ is $O(d\mc{R})$ which is linear in $d$.

\paragraph{Computing the derivative of $\eta_t$}
While the derivative  $\eta_t'$  for the memory term can be easily calculated  for $\eta_t^{\Bayes}$ and $\eta_t^{\marginal}$ via direct differentiation, the derivative for $\eta_t^{\BP}$ is less obvious because  it involves $\mc{R}$ rounds of BP updates \eqref{eq:check_to_var}--\eqref{eq:LLR_last_step}.  Nevertheless, 
using the approach in \cite{amalladinne2022unsourced,ebert2025sparse}, the derivative can be derived in closed form and computed efficiently, provided the number of BP rounds $\mc{R}$ is less than the girth of the LDPC factor graph.
\begin{lem}[Jacobian of $\eta_t^\BP$]
\label{lem:etaBP_deriv}%[Derivatives in Onsager term]  
For $t \ge 1$, consider the AMP decoder with denoiser $\eta_{t}^{\BP}: \reals^d \to \reals^d$,  where  BP is performed for fewer rounds than the girth of the LDPC factor graph.
Let
\begin{align}
\bD:= \frac{\de \eta_{t}^{\BP} (\bs_\ell^t)}{\de \bs_\ell^t}\in\reals^{d\times d}, \ \ \text{ for }  \bs_\ell^t \in \reals^d.
\end{align}
Then  for $j,j_1\in [d]$ and $j\neq j_1,$
\begin{align}
   D_{j,j}  &=  
   \frac{1}{\Sigma^t_{j,j}}\left(E-\left[\eta_t^{\BP}(s_\ell^t)\right]_j^2\right), \label{eq:Jacobian_diag_entries}\\
   D_{j,j_1} &=   \frac{1}{\Sigma^t_{j_1,j_1}}\left(E-\left[\eta_t^{\BP}(s_\ell^t)\right]_j^2\right) \cdot C_{j, j_1},\label{eq:Jacobian_off_diag_entries}
\end{align}
where  $C_{j, j_1}\in (-1,1]$ is a constant that depends on the shortest path between the variable nodes $j$ and $j_1$ on the computation graph of $\mc{R}$ rounds of BP. If there is no path between  nodes $j$ and $j_1$, then $C_{j, j_1} = 0$.
\end{lem}
\begin{proof}
\begin{figure}[t!]
    \centering
\includegraphics[width=0.45\textwidth]{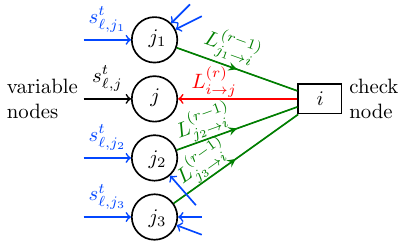}
    \caption{Message {\color{red}$L_{i\to j}^{(r)}$} depends on {\color{green}$\{L_{j' \to i}^{(r-1)}\}$} which in turn depend on {\color{myblue}$\{s_{\ell, j'}^t\}$} and other unlabelled  messages in blue, excluding $s_{\ell, j}^t$. The unlabelled  messages in blue are independent of $s_{\ell, j}^t$ since BP is performed for fewer rounds than girth of the graph.} \label{fig:bipartite}
    \end{figure}
 We will use $j$ and $i$ with subscripts or superscripts to denote variable nodes and check nodes, respectively.     With $\mc{R}$ denoting the number of BP rounds, from \eqref{eq:etat_BP_exp} we have 
    \begin{align}
        D_{j,j}= \frac{\de \left[\eta_{t}^{\BP} (s_\ell^t)\right]_j}{\de s_{\ell, j}^t} 
        &= \frac{\de }{\de L_{j}^{(\mc{R})}}\left(\sqrt{E} \tanh\left(L_j^{(\mc{R})}/2\right) \right) \cdot \frac{\de L_{j}^{(\mc{R})} }{\de s^t_{\ell,j}}\nonumber\\
        & \stackrel{\text{(i)}}{=} \frac{\sqrt{E}}{2}\left(1-\tanh^2\left( L_j^{(\mc{R})}/2\right)\right)\cdot \frac{2\sqrt{E}}{\Sigma^t_{j,j}}\nonumber\\
        & = \left(E-\left[\eta_t^{\BP}(s_\ell^t)\right]_j^2\right) \cdot \frac{1}{\Sigma^t_{j,j}}
    \end{align}
    where (i) uses the fact that
    \begin{equation}
        \frac{\de L_{j}^{(\mc{R})}}{\de s^t_{\ell,j}}  
        \stackrel{\text{(ii)}}{=} \frac{\de L(s_{\ell,j}^t)}{\de s_{\ell,j}^t} \stackrel{\text{(iii)}}{=} \frac{2\sqrt{E}}{\Sigma^t_{j,j}}\,.
        \end{equation}
    The equality (iii) follows from \eqref{eq:initial_LLR}, and (ii) uses  \eqref{eq:LLR_last_step}, noting that each summand $L_{i\to j}^{(\mc{R})}$ is a function of the extrinsic messages $\{L_{j'\to i}^{(\mc{R}-1)} \ \text{for}\ j'\in N(i)\setminus j\}$.  The key observation is that the messages $\{L_{j'\to i}^{(\mc{R}-1)} \ \text{for}\ j'\in N(i)\setminus j\}$ do not depend on  $s_{\ell, j}^t$, since $\mc{R}$ is smaller than the girth of the graph. We illustrate this in Fig.\@ \ref{fig:bipartite}.
    Thus, none of the summands $ L_{i\to j}^{(\mc{R})}$ depends on $s_{\ell,j}^t$ resulting in $\frac{\de }{\de s_{\ell,j}^t}\sum_{i\in N(j)} L_{i\to j}^{(\mc{R})}=0.$  This completes the proof of \eqref{eq:Jacobian_diag_entries}.
% Similarly we can show that  $ D_{j,h}=0$ for $j\neq h$.

    To prove \eqref{eq:Jacobian_off_diag_entries}, we first note that for $j\neq j_1$,
    \begin{align}
        D_{j,j_1}= \frac{\de \left[\eta_t^{\BP}(s_\ell^t)\right]_j}{\de s_{\ell, j_1}^t}
        &= \frac{\sqrt{E} }{2}\left( 1-\tanh^2\left(L_j^{(\mc{R})}/2\right)\right) \cdot \frac{\de L_{j}^{(\mc{R})} }{\de s^t_{\ell,j_1}}.\label{eq:D_j_j1_chain_rule}
    \end{align}
    Consider the computation graph that represents $\mc{R}$ rounds of  BP computation. 
If there is no  path in the graph from variable  node $j$ to variable node $j_1$, then $\frac{\de L_j^{(\mc{R})}}{\de s_{\ell, j_1}^t}=0$ in \eqref{eq:D_j_j1_chain_rule}, leading to  $D_{j, j_1}=0$. Conversely, if  a path does exist from $j$ to $j_1$, then  $\frac{\de L_j^{(\mc{R})}}{\de s^t_{\ell,j_1}}\neq0$ and its value can be computed using the chain rule. To do so, we first note that  \eqref{eq:check_to_var}--\eqref{eq:var_to_check} imply that for any path $j'\to i'\to j''$ in the computation graph,  
\begin{align}
  C_{j'\to i'\to j''}:= \frac{\de L_{{i'}\to j''}^{(r)}}{\de L_{j'\to i'}^{(r-1)}} = \frac{c_{j'\to i'\to j''} \left[1-\tanh^2\left(L_{j'\to i'}^{(r-1)}/2\right)\right] }{1-c_{j'\to i'\to j''}^2\tanh^2\left(L_{j'\to i'}^{(r-1)}/2\right)},\label{eq:chain_rule_vcv}
  % \frac{u}{1-u^2} \frac{1-\tanh^2\left(\frac{L_{j_1 \to i^*}}{2}\right)}{\tanh\left(\frac{L_{j_1 \to i^*}}{2}\right)}
\end{align}
where 
$c_{j'\to i'\to j''}=\prod_{\tj\in N(i')\setminus \{j', j''\}} \tanh \left( L_{\tj\to i'}^{(r-1)}/2\right).$ Since $\tanh u\in (-1,1)$ for any $u\in (-\infty, \infty)$, we have that  $c_{j'\to i'\to j''}\in (-1, 1]$ and so $C_{j'\to i'\to j''} \in(-1, 1]$. Similarly, for any path $i'\to j'\to i'' $ in the computation graph, we have that
\begin{align}
 C_{i' \to j'\to i''}:= \frac{\de L_{j'\to i''}^{(r)}}{\de L_{i'\to j'}^{(r)}}  = 1.\label{eq:chain_rule_cvc}
\end{align}
Consider the general case where the path between $j_1$ and $j$ in the computation graph takes the form  $j_1\to i_1\to j_2 \to  i_2\to\dots \to  j_r \to  i_r\to  j$ with $r\le \mc{R}$. There is at most one such path since $\mathcal{R}$ is smaller than the girth of the graph. 
Then  applying \eqref{eq:chain_rule_vcv}--\eqref{eq:chain_rule_cvc} combined with the chain rule allows us to evaluate $\frac{\de L_j^{(\mc{R})}}{\de s^t_{\ell,j_1}}$ in \eqref{eq:D_j_j1_chain_rule} as follows:
%  then using \eqref{eq:LLR_last_step} we have 
    \begin{align}
        \frac{\de L_j^{(\mc{R})}}{\de s^t_{\ell,j_1}}  &
        % =\frac{\de }{\de (s_\ell^t)_{j_1}}\left[\sum_{i \in N(j)}L_{i\to j}^{(\mc{R})}\right]
        \stackrel{\text{(i)}}{=}  \frac{\de L_{{i_r}\to j}^{(\mc{R})}}{\de s^t_{\ell,j_1}}\nonumber\\
       & =\frac{\de L_{{i_r}\to j}^{(\mc{R})}}{\de L_{j_r\to i_r}^{(\mc{R}-1)}} \cdot\frac{\de L_{j_r\to i_r}^{(\mc{R}-1)}}{\de L^{(\mc{R}-1)}_{i_{r-1}\to j_r}} \cdot\, \dots\, \cdot \frac{\de L_{i_1\to j_2}^{(\mc{R}-r+1)}}{\de L_{j_1\to i_1}^{(\mc{R}-r)}} \cdot \frac{\de L_{j_1\to i_1}^{(\mc{R}-r)}}{\de s^t_{\ell,j_1}}\nonumber\\
        &=  C_{j_r\to i_r\to j}\cdot  C_{i_{r-1}\to j_r\to i_r}\cdot \, \dots \, \cdot C_{j_1\to i_1\to j_2}\cdot \frac{2 \sqrt{E}}{\Sigma_{j_1, j_1}^t}\nonumber\\
        & \stackrel{\text{(ii)}}{=} C_{j, j_1}  \cdot  \frac{2 \sqrt{E}}{\Sigma_{j_1, j_1}^t}, \label{eq:D_j_j1_chain_rule_1}
    \end{align}
    where (i) used \eqref{eq:LLR_last_step}, and (ii) applied $C_{j, j_1}:= C_{j_r\to i_r\to j}C_{i_{r-1}\to j_r\to i_r}\dots C_{j_1 \to i_1 \to j_2} \in (-1, 1]$.  
 Substituting  \eqref{eq:D_j_j1_chain_rule_1}  back into \eqref{eq:D_j_j1_chain_rule}  gives 
 \begin{align}
    D_{j, j_1} = \frac{1}{\Sigma^t_{j_1,j_1}} \left(E-\left[\eta_t^{\BP}(s_\ell^t)\right]_j^2\right) \cdot C_{j, j_1}.
 \end{align}
 This concludes the proof of \eqref{eq:Jacobian_off_diag_entries}.
\end{proof}
Lemma \ref{lem:etaBP_deriv} implies that the characterization of the limiting $\UER$ and $\BER$ in Theorem \ref{thm:UER_BER}  holds for AMP decoding with $\eta^{\BP}_t$ as the denoiser. Formally, we have the following corollary.

\begin{cor} \label{corr:SEvalidity}
The asymptotic guarantees  in Theorem \ref{thm:UER_BER} hold for the AMP decoder with any of the three denoisers: $\eta^{\Bayes}_t, \eta^{\marginal}_t$, and for $\eta^{\BP}_t$ assuming that the number of BP rounds is less than the girth of the LDPC factor graph.
\end{cor}
\begin{proof}
    It can be verified by direct differentiation that the derivatives of 
 $\eta^{\Bayes}_t$ and $\eta^{\marginal}_t$ are bounded. For $\eta_t^{\BP}$, we only need to show that $D_{j,j}$ and $D_{j, j_1}$ in \eqref{eq:Jacobian_diag_entries}--\eqref{eq:Jacobian_off_diag_entries} are  bounded for $j, j_1 \in [d]$ and $j\neq j_1$. This follows by observing that  $\Sigma^t_{j,j} > \sigma^2$ (from \eqref{eq:SE_Sigma_k+1}) and  $[\eta_t^{\BP}(\bs_\ell^t)]_j^2 <E$ (from \eqref{eq:etat_BP_exp}) for any $j\in [d]$.
\end{proof}

\subsection{Numerical results for i.i.d.\@ design} 
\label{sec:iid_numerical}

In this section, we numerically evaluate the tradeoffs achieved by the concatenated coding scheme with an i.i.d.\@ design  using different denoisers. 
For a target $\BER=10^{-4}$, we plot  the maximum spectral efficiency $\S=Lk/n = (L/\tn)(k/d)$  achievable as a function of signal-to-noise ratio $E_b/N_0 = (Ed/k)/(2\sigma^2) $.  We use $\BER$ rather than $\UER$ since the $\UER$ of an uncoded scheme degrades approximately linearly with $d$. For each setting, we also plot the converse bounds from \cite{zadik2019improved}, and the  achievability bounds from either \cite{zadik2019improved} or \cite{hsieh2022near} depending on which one yields the larger achievable region without causing computational issues. These bounds can be used to obtain upper and lower bounds on the  maximum spectral efficiency achievable for given values of $E_b/N_0$  and  $\PUPE$.  
To adapt these bounds to target $\BER$ (rather than target $\PUPE$), we use the random coding assumption that when a  codeword is decoded incorrectly, approximately half of its bits are in error, i.e., $\E[\BER]=\frac{1}{2}\PUPE$.

In Figs.\@ \ref{fig:hamming}--\ref{fig:ldpc_diff_rates}, `SE' refers to curves obtained by using the state evolution result of Theorem  \ref{thm:UER_BER}, and `AMP' (indicated by crosses) refers to points obtained via simulation.  To simplify our implementation, we set the off-diagonal entries of the Jacobian matrix \(\eta_t'(\bs_\ell^{t})\) in \eqref{eq:alg_Zt} to zero. This approximation is justified by Lemma \ref{lem:etaBP_deriv}, which indicates that while the off-diagonal entries of the Jacobian follow a similar form to the diagonal ones, they are scaled by a factor less than one in absolute value. Empirical observations also confirm that the off-diagonal entries are typically much smaller than the diagonal entries. Fig.\@ \ref{fig:hamming} compares the uncoded scheme (i.e., $d=1$) with the concatenated scheme with a $(7,4)$ Hamming code, decoded using AMP with Bayes-optimal denoiser $\eta_t^{\Bayes}$. Even this simple code provides a savings of over 1dB in the minimum $E_b/N_0$ required to achieve positive spectral efficiency, compared to the uncoded scheme as well as the converse bound for $k=1$.

\begin{figure}[t!]
    \centering
    \includegraphics[width=0.7\linewidth]{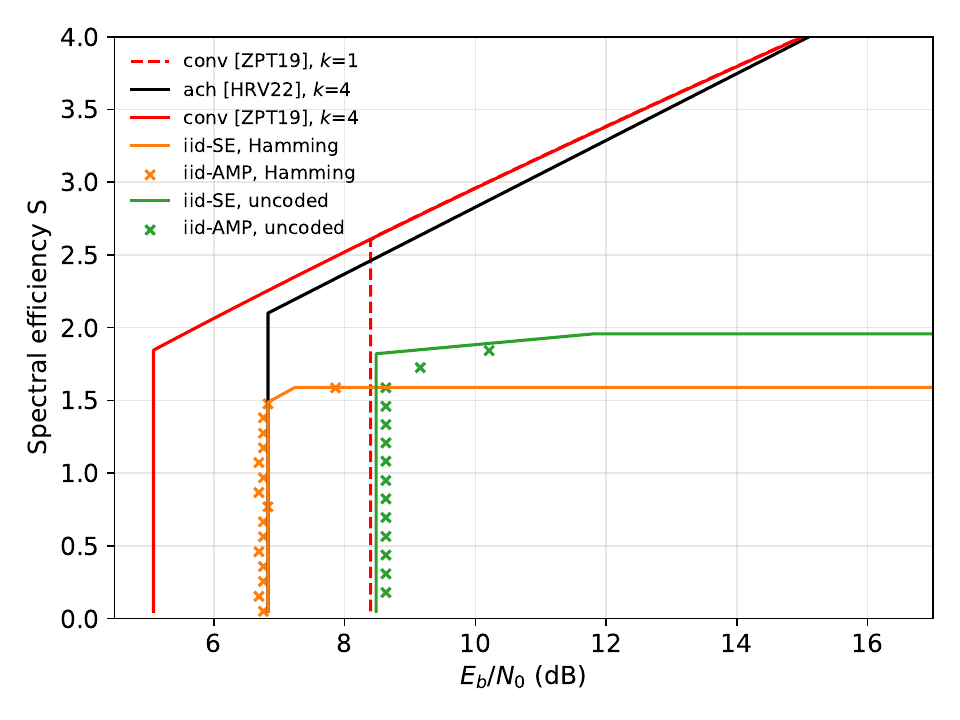}
    \vspace{-0.5cm}
    \caption{Comparison of the uncoded scheme and the concatenated scheme with $(7,4)$ Hamming outer code and  denoiser $\eta_t^\Bayes$. $L=20000.$ }
    \label{fig:hamming} 
    \vspace{-5pt}
    % main_isit_coding_hamm_tight_ach.py
\end{figure}

Figs.\@ \ref{fig:bp_v_marginal} and \ref{fig:ldpc_diff_rates} employ LDPC codes  from the IEEE 802.11n  standards as outer codes (with codelength $d=720$ bits).  Fig.\@ \ref{fig:bp_v_marginal} considers a user payload of $k=360$ bits and  compares the decoding performance of AMP with different denoisers: the marginal-MMSE $\eta_t^\marginal$ or the BP denoiser $\eta_t^\BP$ which executes 5 rounds of BP per AMP denoising step. The latter outperforms the former by around 7.5dB since $\eta_t^\marginal$ does not use the parity constraints of the code. The dotted orange curve  in Fig.\@ \ref{fig:bp_v_marginal} shows that the performance of AMP with $\eta_t^\marginal$ is substantially improved by running a BP decoder (200 rounds) after AMP has converged. This additional BP decoding at the end also improves the performance of $\eta_t^\BP$ (blue dotted curve). We observe that the achievable spectral efficiency with $\eta_t^\BP +$ BP is consistently about $40\%$ higher than with $\eta_t^\marginal+$ BP.

\begin{figure}[t!]
    \centering
    \includegraphics[width=0.7\linewidth]{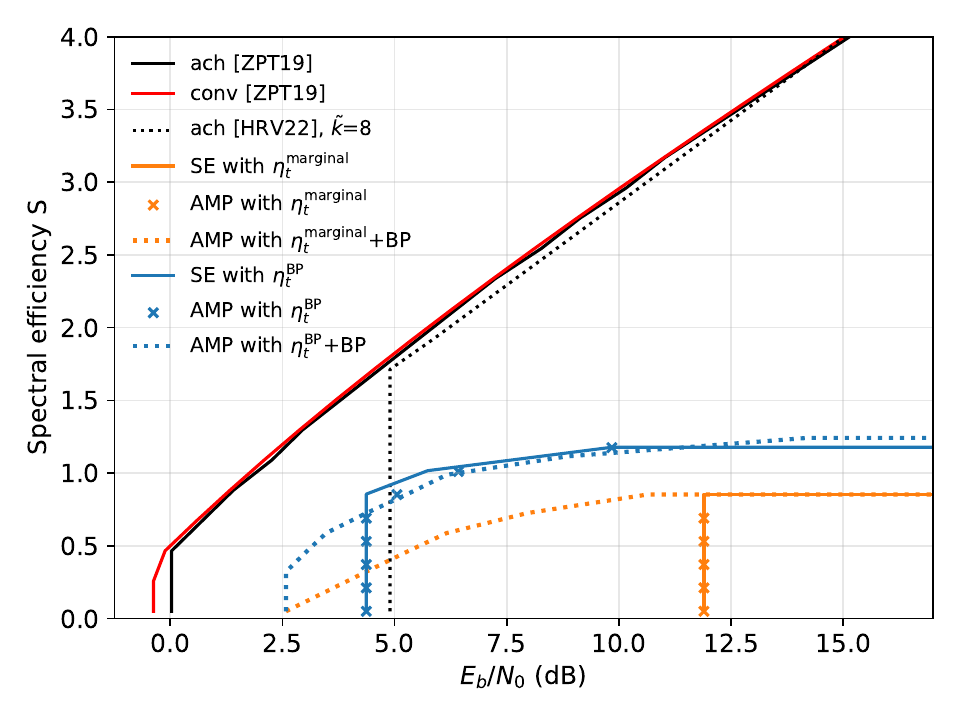}
    \vspace{-0.5cm}
    \caption{Comparison of marginal-MMSE denoiser $\eta_t^{\marginal}$ (orange) and BP denoiser $\eta_t^{\BP}$ (blue) for decoding LDPC outer code (with fixed rate 1/2).  The dotted orange and blue plots correspond to AMP decoding coupled with 200 rounds of BP after AMP has converged. The dotted black curve corresponds to the  SPARC-based  scheme of \cite{hsieh2022near} with $\tilde{k}=8$.  $L=2000, k=360, d=720$. }
    \label{fig:bp_v_marginal}
        \vspace{-7pt}
\end{figure}
\begin{figure}[t!]
    \centering
    \includegraphics[width=0.7\linewidth]{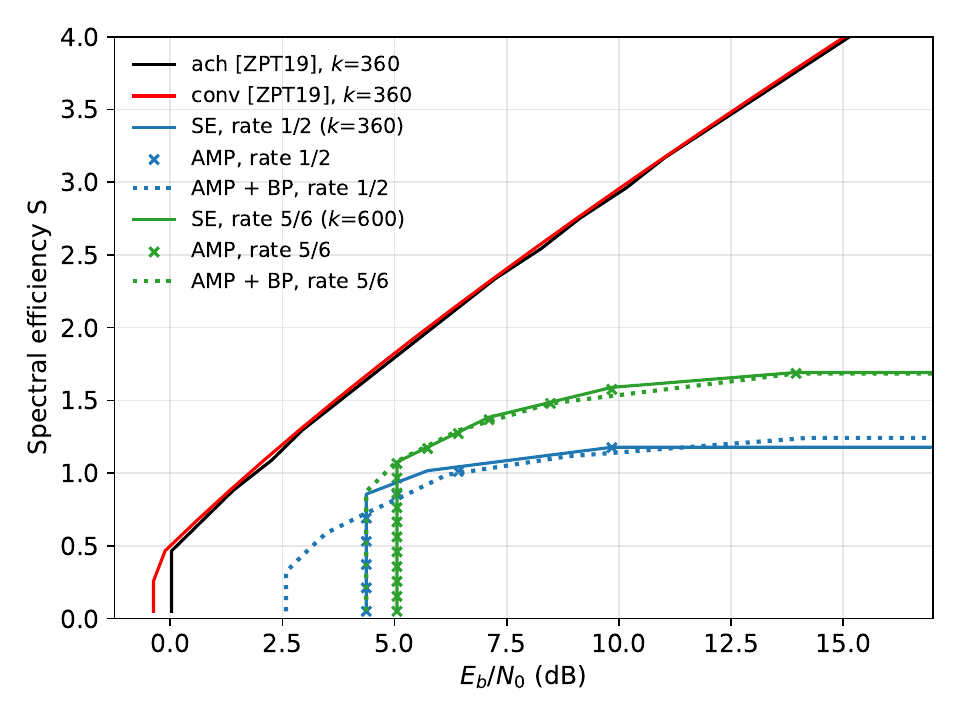}
    \vspace{-0.5cm}
    \caption{Comparison of LDPC outer code with rate 1/2 (blue) or 5/6 (green)  under AMP decoding  with BP denoiser $\eta_t^{\BP}$. Dotted curves  correspond to AMP decoding coupled with 200 rounds of BP after AMP has converged. $L=2000, d=720$. 
    The achievability and converse bounds for $k=600$ are omitted as they  nearly match those for $k=360$.}
    \label{fig:ldpc_diff_rates}
    \vspace{-3pt}
\end{figure}
 %The Python code used to produce the figures  is available at \url{https://github.com/ShirleyLiuXQ/gmac\_amp\_public}.

The dotted black curve in Fig.\@ \ref{fig:bp_v_marginal} corresponds to  the asymptotic performance of the SPARC-based scheme  of \cite{hsieh2022near} with a spatially coupled design and  $\tilde{k}=8$, the highest payload the scheme  could handle in our simulations.
Hence, to transmit a 360-bit payload,  the SPARC-based scheme needs to be used  45  times, with 8 bits transmitted each time. We observe that for smaller spectral efficiency, the SPARC-based scheme is outperformed  by our  concatenated scheme with BP post-processing (dotted blue curve) by  about 2.5dB.

Fig.\@ \ref{fig:ldpc_diff_rates} compares the performance of the concatenated scheme with  LDPC codes with different rates: 1/2 and 5/6. The AMP denoiser is $\eta_t^\BP$, and the dotted curves show the effect of adding BP decoding (200 rounds) after AMP convergence.
The code with the higher rate $5/6$ achieves higher spectral efficiency for large values of $E_b/N_0$, but the rate $1/2$ code achieves a positive spectral efficiency for smaller $E_b/N_0$ values, down to $2.5$dB. We expect that using  an outer LDPC code with a  rate  lower than 1/2 will enable the concatenated scheme to achieve a  positive spectral efficiency at  $E_b/N_0$ values even  below 2.5dB.

In Figs.\@ \ref{fig:hamming}--\ref{fig:ldpc_diff_rates}, the asymptotic performance of AMP, predicted by state evolution, closely tracks its actual performance at large, finite $n, L$ (with the simplification where the off-diagonal entries of the Jacobian are  set to zero). Moreover, 
 considering a metropolitan area with $10^6$ to $10^7$ devices and each device active a few times per hour, the user density $\mu$  is typically $10^{-4}$ to $10^{-3}$ \cite[Remark 3]{zadik2019improved}. For user densities in this range and per-user payload $k$ on the order of $10^2$ to $10^3$, the spectral efficiency $\S = \mu k $ is less than 1. In all figures,  the concatenated coding schemes exhibit the most substantial improvements for $\S <1$.

  The concurrent work  by Ebert et al. \cite{ebert2024multiusersrldpccodes} proposes a scheme for the GMAC using concatenated SPARC-LDPC codes with AMP-BP decoding. Since their construction is based on   SPARC, the memory and the decoding complexity scales exponentially with the section size used for the SPARC. Therefore,  the decoder can handle only a small number of users if the per-user payload is large.  Fig. 4 in \cite{ebert2024multiusersrldpccodes}, which reports the performance of the SPARC-LDPC scheme for $L \le 16$  users each with a payload of $k=584$ bits, shows that a spectral efficiency up to $\S\approx 0.94$ can be achieved at $E_b/N_0=3$dB, for a target $\BER=10^{-4}$. This is  slightly higher than the achievable spectral efficiency in Fig.\@ \ref{fig:ldpc_diff_rates}. However, we emphasize that our scheme is implemented for $L=2000$ users, each with a payload of $360$ bits. It appears challenging to scale the SPARC-LDPC scheme to such a large number of users.

As illustrated in Figs.\@ \ref{fig:hamming}--\ref{fig:ldpc_diff_rates},  the gap between the spectral efficiency achieved by our concatenated scheme and the converse bounds grows with $E_b/N_0$. In the next section,  we demonstrate how the spectral efficiency of our scheme can be substantially improved by using a spatially coupled design matrix \cite{hsieh2022near}.

\section{Spatially coupled  design and AMP decoder}\label{sec:SC_AMP}

In this section, we study the concatenated scheme with a spatially coupled Gaussian design matrix $\bA$. We first define the spatial coupled design,  and  then describe the corresponding AMP decoder and characterize its asymptotic performance (Theorem \ref{thm:UER_BER_sc}). 

A spatially coupled (SC) design  matrix  $\bA \in \reals^{\tn \times L}$ is divided into $\R \times \C$ equally sized blocks. The entries  of  $\bA$ within each block are i.i.d.\@ Gaussian with zero-mean and variance   prescribed by a \emph{base matrix} $\bW\in \reals^{\R \times \C}$. Specifically, 
the matrix $\bA$ is obtained  by replacing each entry  $W_{\sfr,\sfc}$ of the base matrix
by an $\frac{\tn}{\Lr} \times \frac{L}{\Lc}$ block with entries drawn  $\stackrel{\text{i.i.d.}}{\sim} \mc{N}(0, \frac{1}{\tn/\Lr} W_{\sfr,\sfc})$, 
for $\sfr\in[\Lr]$, $\sfc\in[\Lc]$. That is, we have
\begin{equation}
    \label{eq:construct_sc_A}
A_{i\ell}\stackrel{\text{i.i.d.}}{\sim} \mc{N}\bigg(0,\frac{1}{\tn/\Lr} 
W_{\sfr(i), \sfc(\ell)} \bigg),  \quad \text{for }  i \in [\tn], \ \ell\in [L].
\end{equation}
Here, the operators $\sfr(\cdot):[\tn]\rightarrow[\Lr]$ and $\sfc(\cdot):[L]\rightarrow[\Lc]$  map a particular row or column index of $\bA$ to its corresponding \emph{row block} or \emph{column block} index of $\bW$.  See Fig.\@ \ref{fig:spatial_coupling_example} for an example.

\begin{figure}[t!]
    \centering
\includegraphics[width=0.5\linewidth]{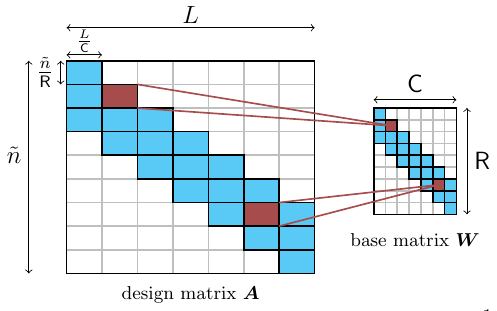}
    \vspace{-0.1cm}
    \caption{A spatially coupled design matrix $\bA$ constructed using a base matrix $\bW$ according to  \eqref{eq:construct_sc_A}. The base matrix shown here is an $(\omega=3, \Lambda=7)$ base matrix (defined in Definition \ref{def:ome_lamb_rho}). The white parts of $\bA$ and $\bW$ correspond to zeros.}
    \label{fig:spatial_coupling_example}
    \end{figure}

 As in \cite{donoho2013information}, the entries of the base matrix $\bW \in \reals_{+}^{\R\times \C}$ are scaled to satisfy: 
$   \sum_{\sfr=1}^{\Lr} W_{\sfr, \sfc} = 1 \; \text{ for } \sfc \in [\Lc]$.
 This is to ensure that the columns of $\bA$ (the signature sequences $\ba_\ell$) have unit squared $\ell_2$-norm in expectation.  Each block of  $\bA$ can be viewed as an (uncoupled) i.i.d.\@ design matrix with ${L}/{\C}$ users and $d{\tn}/{\R} = {n}/{\R}$  channel uses. Thus, we define the \emph{inner} user density as 
\begin{align}
    \muin := \frac{L/\C}{n/\R} = \frac{\R}{\C}\mu .
\end{align}
The standard i.i.d.\@ Gaussian design where  $A_{i\ell}\stackrel{\text{i.i.d.}}{\sim}\mc{N}(0, 1/\tn)$ is  a special case of the SC  design, obtained by using a base matrix with a single entry ($\Lr=\Lc=1$).
We shall use a class of base matrices called  $(\omega, \Lambda)$ base matrices that have also been  used for spatially coupled sparse superposition codes   \cite{rush2021capacity,hsieh2022near}.

\begin{definition}
    \label{def:ome_lamb_rho}
    An $(\omega , \Lambda)$ base matrix $\bW$ is described by two parameters: the coupling width $\omega\geq1$ and the coupling length $\Lambda\geq 2\omega-1$. The matrix has $\Lr=\Lambda+\omega-1$ rows and $\Lc=\Lambda$ columns, with each column having $\omega$ identical nonzero entries in the band-diagonal and zeros everywhere else.  For $\sfr \in [\Lr]$ and $ \sfc\in[\Lc]$, the $(\sfr,\sfc)$th entry of the base matrix is given by
    \begin{equation}
    \label{eq:W_rc}
    W_{\sfr,\sfc} =
    \begin{cases}
            \ \frac{1}{\omega} \quad &\text{if} \ \sfc \leq \sfr \leq \sfc+\omega-1\,,\\
        \ 0 \quad &\text{otherwise}\,.
    \end{cases}
    \end{equation}
    \end{definition}

Fig.\@ \ref{fig:spatial_coupling_example} shows an example of design matrix constructed using an $(\omega=3, \Lambda=7)$ base matrix. From the GMAC perspective, the main difference between such a  design and an i.i.d.\@ Gaussian one is that   only a small fraction of users are active during any given channel use. In Fig.\@ \ref{fig:spatial_coupling_example}, only the first $L/\C$ users are active during the first $\tn/\R$ channel uses, and only the last $L/\C$ users are active for the last $\tn/\R$ channel uses. This allows these two sets of users to be decoded more easily than the others, thus helping the  decoding of the adjacent blocks of users. This creates a decoding wave that propagate from the ends to the center. Examples and figures illustrating the decoding wave in spatially coupled designs can be found in \cite{rush2021capacity,krzakala2012statistical}.

\subsection{Spatially Coupled AMP}
The AMP decoder for a spatially coupled Gaussian design matrix (SC-AMP) is a generalization of the AMP decoder for the i.i.d.\@ Gaussian design  presented in Section \ref{sec:iid_AMP}, accounting for the fact that the SC design has an $\R \times \C$ block-wise structure, with potentially different variances across blocks.
%The update rules encompass  terms and functions that are  associated with  the different blocks of the design matrix, and  are defined differently for the different blocks.
Starting with initialization $\bX^0=\bzero_{L\times d}$ and  $\btZ^0 = \bzero_{\tn \times d}$, 
% \RV{added initialization}
  the  decoder computes for  $t\ge 0$:
\begin{align}
&\bZ^t = \bY -\bA\bX^t +\btZ^t,\label{eq:alg_sc_Zt}\\
&\bX^{t+1} =\eta_t\left( \bS^t\right), \quad \text{where}\quad \bS^t=\bX^t+\bV^t.
\label{eq:alg_sc_Xt}    
\end{align}
Here the denoising function $\eta_t: \mathbb{R}^{L\times d}\to \mathbb{R}^{L\times d}$ is assumed to be Lipschitz and acts  row-wise on matrix inputs. Denoting the rows of $\bS^t$ by $\bs_\ell^t$ for $\ell \in [L]$,  we have:
\begin{align}
    \eta_t(\bS^t) = \begin{bmatrix}
        \eta_{t,1}\left(\bs^t_1\right)\\
        \vdots\\
        \eta_{t,1}\left(\bs^t_{L/\C}\right)\\
        \vdots\\
        \eta_{t,\C}\left(\bs^t_{(\C-1)L/\C+1}\right)\\
        \vdots\\
        \eta_{t,\C}\left(\bs^t_{L}\right)
    \end{bmatrix}
    \setlength{\arraycolsep}{0pt} % Avoid any column space in arrays that follow
    \begin{array}{ c }
      \left.\kern-\nulldelimiterspace
      \vphantom{\begin{array}{ c }
          \eta_1  \\ 
        \vdots \\ 
          \eta_1 \\   
        \end{array}}
        \right\}\text{$\frac{L}{\C}$ rows with $\sfc=1$}\\\\
        \vphantom{\vdots} % First row
        \left.\kern-\nulldelimiterspace
        \vphantom{\begin{array}{ c }
          \eta_\C \\
          \vdots\\
            \eta_\C\\
      \end{array}}
      \right\}\text{$\frac{L}{\C}$ rows with $\sfc=\C$\,,}
    \end{array}
        \label{eq:eta_t_separable}
\end{align}
where $\eta_{t,\sfc}: \mathbb{R}^{d}\to \mathbb{R}^{d}$ corresponds to the denoising function applied to  the $\frac{L}{\C}$ rows  in block $\sfc\in [\C]$. We have used the convention that the function returns a row vector when applied to a row vector.

For $t\ge 0$,  $\tilde{\bZ}^t$ and $ \bV^t$  in \eqref{eq:alg_sc_Zt}--\eqref{eq:alg_sc_Xt} are defined through a matrix 
$\boldsymbol{Q}^t\in\reals^{d\R\times d\C}$, which consists of $ \R \times \C$   submatrices, each of size $d\times d$.  For $\sfr \in [\R], \sfc \in [\C]$, the submatrix $\bQ_{\sfr,\sfc}^t \in \reals^{d\times d}$ is defined as:
\begin{equation}
\bQ_{\sfr,\sfc}^t=\left[\bPhi_{\sfr}^t\right]^{-1}\bT_\sfc^t\,,
    \label{eq:Q_t_def}
\end{equation}
where $\bPhi_\sfr^t, \bT_\sfc^t\in\reals^{d\times d}$ are  deterministic matrices defined later  in \eqref{eq:SC_SE_Phi_t}--\eqref{eq:G_c_t}, as part of the corresponding state evolution. 
The $i$th row of matrix $\btZ^t\in\reals^{\tn \times d}$ then takes the form:
\begin{equation}
    \btz_i^t = {d\muin} \, \bz_i^{t-1} \,  \sum_{\sfc =1}^{\Lc} W_{\sfr(i), \sfc}\,\bQ^{t-1}_{\sfr(i), \sfc} \,  
    \frac{1}{L/\C}\cdot \sum_{\ell\in\mc{L}_\sfc} \left[\eta_{t-1, \sfc(\ell)}'( \bs_\ell^{t-1} )\right]^\top,\quad \text{for}\; i\in[\tn],
    \label{eq:Z_tilde}
\end{equation}
where
 $\mc{L}_\sfc= \{(\sfc-1)L/\C+1, \dots, \sfc L/\C\}$, and 
$\eta_{t,\sfc}'(\bs) = \frac{\de \eta_{t,\sfc}(\bs)}{\de \bs}\in\reals^{d\times d}$ is the Jacobian of $\eta_{t,\sfc}$. Quantities with negative  iteration index are set to all-zero matrices. The $\ell$th row of $\bV^t\in \reals^{L\times d}$ takes the form:
\begin{equation}
    \bv_\ell^t = \sum_{i=1}^{\tn} A_{i\ell} \, \bz_i^t \, \bQ_{\sfr(i),\sfc(\ell)}^t\,, \quad \text{for}\;\ell\in[L]
    \label{eq:SC_AMP_V_t}.
\end{equation}
In \eqref{eq:Z_tilde} and \eqref{eq:SC_AMP_V_t}, the vectors $\tilde{\bz}_i^t, \bz^{t}_i$ and $ \bv_\ell^t$ are all row vectors.

\paragraph{Spatially Coupled State Evolution (SC-SE)}\label{sec:SC-SE} Similar to the  i.i.d.\@ case in Section \ref{sec:iid_AMP}, we now state the asymptotic distributional guarantees for the AMP iterates  with the SC  design.  For each iteration $t\ge 1,$ in the limit as $L, n\to \infty$ with $L/n\to \mu $,  the memory term $\tilde{\bZ}^t$ in \eqref{eq:alg_sc_Zt} ensures that 
the  empirical distribution  of the rows of  $\bZ^t$ in the  block $\sfr\in [\R]$ converges  to a Gaussian  $\mc{N}_d(\bzero, \bPhi_{\sfr}^t) $. Furthermore, the empirical distribution of the  rows of $(\bS^t-\bX)$ in block $\sfc\in [\C] $ converges to another Gaussian $\mc{N}_d(\bzero, \bT_{\sfc}^t)$. The covariance matrices $ \bPhi_{\sfr}^t, \bT_{\sfc}^t\in \reals^{d\times d}$ are iteratively defined via  the spatially coupled state evolution (SC-SE), a deterministic recursion defined as follows. Starting with initialization with $\bPsi_{\sfc}^0 =E\bI_{d}$,  for $t\ge 0 $  and $\sfr\in[\Lr]$, $\sfc\in[\Lc]$ we define:
\begin{align}
&\bPhi_{\sfr}^t  =  \sigma^2 \bI_{d} \, + \,  d\muin\sum_{c=1}^\Lc W_{\sfr,\sfc}\bPsi_\sfc^t\,,\label{eq:SC_SE_Phi_t}\\
&\bPsi_\sfc^{t+1} =  \E \left\lbrace\left[ \eta_{t, \sfc}\left(\bbx +\bg_{\sfc}^{t}\right)-\bbx\right]\left[ \eta_{t, \sfc}\left(\bbx +\bg_{\sfc}^{t}\right)-\bbx\right]^\top\right\rbrace \label{eq:SC_SE_Psi_t},\\
& \qquad \qquad \text{where } \bg^t_{\sfc}\sim\mathcal{N}_d(\bzero, \bT_{\sfc}^t), \quad   \bT_{\sfc}^t ={\left[\sum_{r=1}^\Lr W_{\sfr,\sfc}{[\bPhi_\sfr^t]}^{ -1}\right]^{-1}}. 
\label{eq:G_c_t}
\end{align}
Here  $\bg_{\sfc}^t$ is independent of $\bbx \sim p_{\bbx}$. 
We can interpret $\bPsi_\sfc^t$ as the asymptotic covariance of the error in the estimated codeword $\bx^t_\ell$ relative to the true codeword $\bx_\ell$, for $\ell\in\mc{L}_\sfc$, i.e., for the $\sfc$-th block of users. The asymptotic error rates of AMP decoding with SC design  are characterized by the following theorem.

\begin{thm}[Asymptotic $\UER$ and $\BER$ with SC design]\label{thm:UER_BER_sc} 
Consider  the concatenated scheme with a spatially coupled Gaussian design, with the assumptions in Section \ref{subsec:assumptions}, and the AMP decoding algorithm in \eqref{eq:alg_sc_Zt}--\eqref{eq:alg_sc_Xt} with  Lipschitz continuous denoisers $\eta_t:\reals^{L\times d} \to \reals^{L\times d}$, for $t \ge 1$. Let $\hat{\bx}^{t+1}_\ell = h_{t,\sfc(\ell)}(\bs^t_\ell)$ be the hard-decision estimate of $\bx_\ell$ 
in iteration $t$.
Then the asymptotic $\UER$ and $\BER$ in iteration $t$    satisfy the following almost surely, for $t \ge 0$:
\begin{align}\label{eq:UER_sc_thm}
        \lim_{L\to \infty} \UER &  := 
        \lim_{L\to \infty} \frac{1}{L}\sum_{\ell=1}^L\ind\left\lbrace\hat{\bx}^{t+1}_\ell \neq \bx_\ell\right\rbrace  
        = \frac{1}{\C}\sum_{\sfc=1}^\C\prob \left( h_{t, \sfc}\left(\bar{\bx}+ \bg_\sfc^{t}\right) \neq \bar{\bx} \right),  \\
    \label{eq:BER_sc_thm} 
    \lim_{L\to \infty} \BER &:= 
        \lim_{L\to \infty}\frac{1}{Ld}\sum_{\ell=1}^L\sum_{j=1}^d \ind\left\lbrace \hat{x}^{t+1}_{\ell,j}\neq x_{\ell,j}\right\rbrace  
        = \frac{1}{\C}\sum_{\sfc=1}^\C\frac{1}{d}\sum_{j=1}^d\prob \left( \left[h_{t, \sfc}\left(\bar{x}+g_\sfc^{t}\right)\right]_j\neq \bar{x}_{j}\right).
\end{align}
Here $\bar{\bx} \sim p_{\bar{\bx}}$  and $\bg_\sfc^t\sim \normal_d(\bzero, \bT_\sfc^t)$ are independent, with $\bT_\sfc^t$ defined  in \eqref{eq:G_c_t}. The limit is taken as $n,L\to \infty$ with $L/n\to \mu$.
\end{thm}
\begin{proof}
    See  Section \ref{sec:proof_sc}. 
\end{proof}

\paragraph*{Choice of $\eta_{t}$ and $h_{t}$ in SC-AMP}
The denoiser $\eta_t$ in the AMP algorithm can be chosen analogously to the i.i.d.\@ case in Section \ref{sec:iid_denoisers}, except that it now acts in a block-dependent manner as shown in \eqref{eq:eta_t_separable}. Take the Bayes-optimal denoiser $\eta_t^{\Bayes}$ as an example:
since the empirical distribution of the  rows of $\bS^t$ in  block $\sfc$ converges to  the law of  $\bar{\bx}+\bg_{\sfc}^t$, with $\bg_\sfc^t\sim\normal_d(\bzero, \bT_\sfc^t)$, we define for $\sfc \in [\C]$ and $\bs\in \reals^{d}$, 
\begin{align}
   &\eta_{t,\sfc}^{\Bayes}(\bs)  =\mathbb{E} \left[\bbx \ |\  \bbx + \bg_{\sfc}^t=\bs \right] 
   = \sum_{\bx' \in \mc{X} }\bx'\cdot \frac{   \exp \left( -\frac{1}{2}(\bx' - 2\bs )^\top \left(\bT_\sfc^t\right)^{-1} \bx' \right)}{  \sum_{\tilde{\bx}'\in \mc{X}} \exp \left( -\frac{1}{2}(\tilde{\bx}' - 2\bs)^\top \left(\bT_\sfc^t\right)^{-1}\tilde{\bx}'\right)}\, .\label{eq:Bayes_denoiser_sc}
\end{align} 
The marginal-MMSE denoiser $\eta_t^{\marginal}$ and the BP denoiser $\eta_t^{\BP}$ can be defined similarly, with the only difference from the i.i.d.\@ versions being that the covariance matrices are block-dependent for the spatially coupled case. The hard-decision estimator $h_{t}$ is also be defined analogously, with  block-dependent covariance. For example, the MAP hard-decision estimate takes the form 
\begin{align}
    h_{t,\sfc}(\bs)
   &= \argmax_{\bx'\in \mc{X} } \prob\left({\bbx} = \bx' \mid \bbx+ \bg_{\sfc}^t =\bs\right).\label{eq:MAP_est_sc}
\end{align}
Observe that \eqref{eq:Bayes_denoiser_sc}--\eqref{eq:MAP_est_sc} are the counterparts of \eqref{eq:Bayes_denoiser}--\eqref{eq:MAP_est}.

\subsection{Numerical results for spatially coupled  design}
\label{sec:SC_numerical}

\begin{figure}[t!]
    \centering
    \includegraphics[width=0.7\linewidth]{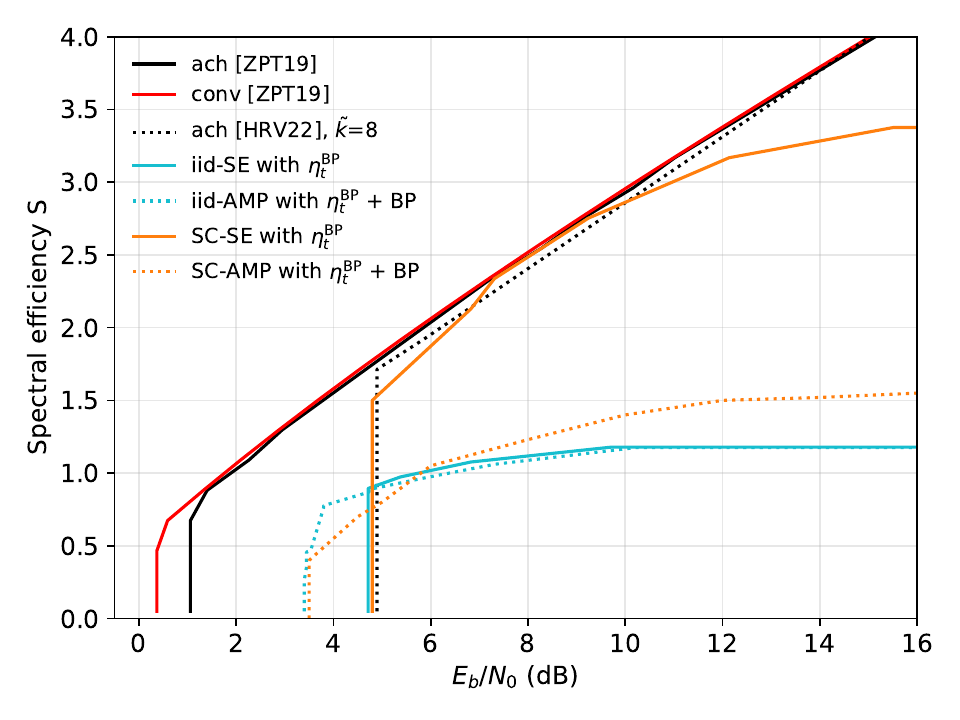}
    \vspace{-0.5cm}
    \caption{Performance of the concatenated scheme with i.i.d.\@ design (cyan) vs.\@ SC design (orange). Dotted cyan and orange curves correspond to AMP with $\eta_t^\BP$ coupled with 200 rounds of BP after AMP has converged. Solid cyan and orange curves are obtained using state evolution. Outer  LDPC code has rate 1/2,    $ k=120$, and $L=2000$ for i.i.d.\@ design; $L=40000$ and $(\omega=4, \Lambda=20)$ for SC design.} \label{fig:ldpc_iid_v_sc_z10}
    \vspace{-3pt}
\end{figure}
In Fig.\@ \ref{fig:ldpc_iid_v_sc_z10}, 
we numerically compare the  error rates of the concatenated scheme with the i.i.d.\@ Gaussian design and the SC Gaussian design  constructed using an $(\omega=4, \Lambda=20)$ base matrix.  We consider the user payload  $k=120$,  a rate 1/2 outer LDPC code, and  the  target $\BER=10^{-4}$. 
 The solid  black and red curves in Fig.\@ \ref{fig:ldpc_iid_v_sc_z10} correspond to the achievability and converse bounds of \cite{zadik2019improved}.  The  `iid-SE' and `SC-SE' plots are obtained using the asymptotic error  characterization of our concatenated scheme (as $n,L\to\infty$ with $L/n\to \mu$), provided in Theorems \ref{thm:UER_BER} and \ref{thm:UER_BER_sc} for i.i.d.\@ and  SC Gaussian design matrices (with $\omega=4, \Lambda=20$), respectively. The dotted black curve is the asymptotic performance of the SPARC-based scheme  of \cite{hsieh2022near} with a spatially coupled design (as $n,L\to \infty$ with $L/n\to \mu$, and  $\omega, \Lambda\to \infty$ with $\Lambda/\omega \to \infty$). As in  Fig. \ref{fig:bp_v_marginal}, the highest payload the AMP decoder for SPARC-based scheme can handle in simulations is $\tilde{k}=8$. 
Hence, to transmit a 120-bit payload,  the scheme needs to be used  15  times, with 8 bits transmitted each time.

 We observe that  for higher values of $E_b/N_0$, the spectral efficiency asymptotically achievable by  the concatenated scheme with the spatially coupled design (`SC-SE') is significantly higher than with the i.i.d.\@ design (`iid-SE').  The  `SC-SE' plot plateaus at  $E_b/N_0>12$dB, unlike the SPARC-based scheme of \cite{hsieh2022near} (dotted black), since our results are derived for fixed, relatively small values of $\omega, \Lambda$ ($\omega=4, \Lambda=20$), while the performance characterization in \cite{hsieh2022near} is for $\omega, \Lambda\to \infty$ with $\Lambda/\omega \to \infty$.

 The `iid-AMP' and `SC-AMP' plots   illustrate the simulated  performance of the AMP decoder with the BP denoiser $\eta_t^{\BP}$ with i.i.d.\@ or SC  design, followed by 200 rounds of BP. The number of users considered is $L=2000$ for the i.i.d.\@ case and $L=40000$ for the SC case. The gap between `SC-AMP' and `SC-SE' is due to the finite-length effects, similar to the observations made in \cite{hsieh2022near}. We observe that for lower values of $E_b/N_0$, the achievable region of the concatenated scheme  with BP post-processing, marked by `iid-AMP' and `SC-AMP', is noticeably larger than that of the near-optimal scheme of \cite{hsieh2022near} (dotted black). The experiments presented in Fig.\@ \ref{fig:ldpc_iid_v_sc_z10} used a relatively  short LDPC code. The performance gains  at small $E_b/N_0$ could be substantially enhanced by using a longer LDPC code in the concatenated scheme (see Fig.\@ \ref{fig:bp_v_marginal}).

 The Python code to reproduce all the numerical results is available at our public GitHub repository \cite{liu_amp_github}.

\section{Implementation details}
%\SL{Since this subsection discusses results for the iid case too, should we make it a standalone section, e.g. Section 6? } \RV{done.}
We discuss a few  implementation details for the numerical results in Sections \ref{sec:iid_numerical} and \ref{sec:SC_numerical}.

\paragraph{Covariance estimation}
   In our  iid-AMP simulations,  the state evolution parameter $\bSigma^t\in\mathbb{R}^{d \times d}$ in \eqref{eq:SE_Sigma_k+1}, used by the Bayes-optimal denoiser $\eta_t^\Bayes$ in \eqref{eq:Bayes_denoiser} and the MAP hard-decision estimator $h_t$ in \eqref{eq:MAP_est}, is  estimated  using the modified residual term $\bZ^t$ in \eqref{eq:alg_Zt} via:
\begin{equation}
    \hat{\bSigma}^t = \frac{1}{\tn}\sum_{i\in[\tn]} {\bz_i^t}({\bz_i^t})^\top,\qquad \text{for}\; t\ge 1.
\end{equation}
State evolution results (not stated here) guarantee that for $t \ge 1$, we have $\hat{\bSigma}^t \to \bSigma^t$ almost surely as $L \to \infty$ with $L/n \to \mu$. 
Recall that $d = n / \tn$ does not grow with $n$ and $L$. Since the AMP decoder uses $\eta_t^\Bayes$ and the MAP estimator $h_t$ only for small $d$ (e.g., $d=7$ for the Hamming code in Fig.\@ \ref{fig:hamming}),  a sufficiently accurate estimate  $\hat{\bSigma}^t$ can be obtained with moderately large $\tn$ and $L$.

For larger $d$ (e.g., $d=720$ for the LDPC  code used in Figs. \ref{fig:bp_v_marginal} and \ref{fig:ldpc_diff_rates}), the AMP decoder employs $\eta_t^\marginal$ or $\eta_t^\BP$ which do not require the full covariance matrix $\bSigma^t$, but only the diagonal entries $\Sigma^t_{j,j}$ for $j\in[d]$. We  estimate these entries via:
\begin{equation}
    \hat{\Sigma}^t_{j,j} = \frac{1}{\tn}\sum_{i\in[\tn]} \left(\bz_i^t\right)_j^2,\qquad \text{for}\; t\ge 1, \; j\in[d].
\end{equation}
Estimating the diagonal entries alone is much less computationally expensive than the full $d\times d$ covariance matrix, and  ensures the  computational cost of the AMP decoder with $\eta_t^\marginal$ or $\eta_t^\BP$ denoiser remain linear in $d$ (and $k$). 

Similarly, in our SC-AMP simulations, the SC-SE parameter $\bT_\sfc^t \in \reals^{d\times d}$ in \eqref{eq:G_c_t} for $\sfc\in[\C]$ is estimated consistently via
\begin{equation}
    \hat{\bT}_\sfc^t= \frac{1}{\tn/\R}\sum_{i\in\mc{I}_{\sfr}} {\bz_i^t} (\bz_i^t)^\top, \qquad \text{for}\; t\ge 1,
\end{equation}
where $\mc{I}_\sfr=\{(\sfr-1)\tn/\R+1, \dots, \sfr\tn/\R\}$.

\paragraph{Estimation of state evolution MSE} 
In each iteration of the iid-SE recursion  \eqref{eq:SE_Sigma_k+1}, the MSE term $\E\{[ \eta_{t}(\bar{\bx} +\bg^{t})-\bar{\bx}] [\eta_{t}(\bar{\bx} +\bg^{t})-\bar{\bx}]^\top\}$  is  estimated by generating  Monte Carlo samples of $\bg^t\stackrel{\text{i.i.d.}}{\sim}\normal_d(\bzero, \bSigma^t)$ one randomly chosen sample of $\bbx$, and applying  $\eta_t$ to each sample of $\bar{\bx}+\bg^t$. Due to the linearity of the outer code, this MSE term is independent of the chosen  sample of $\bbx$, so we can fix $\bbx$ to be the all-zero codeword across all iterations. For consistently estimating the $d \times d$ MSE matrix, the number of  samples of $\bg^t$  needs to be much larger than $d$.
Similarly, in the spatially coupled case,  for each $\sfc\in[\C]$, the MSE term $\E\{[ \eta_{t,\sfc}(\bbx +\bg_{\sfc}^{t})-\bbx] [\eta_{t,\sfc}(\bbx +\bg_{\sfc}^{t})-\bbx]^\top\}$ in \eqref{eq:SC_SE_Psi_t} is estimated using many Monte Carlo samples of $\bg_{\sfc}^t\stackrel{\text{i.i.d.}}{\sim}\normal_d(\bzero, \bT_{\sfc}^t)$  and one randomly chosen sample of $\bbx$. 
With the $\eta_t^\BP$ denoiser, the computational and memory cost of MSE estimation  could potentially be reduced using population dynamics \cite{Mezard2009information}.

\section{Proofs of theorems} \label{sec:proof_UER_BER}
\subsection{Proof of Theorem \ref{thm:UER_BER}}
\label{sec:proof_iid}
The proof is similar to the proof of Theorem 1 in \cite{hsieh2022near}. We begin with a state evolution characterization of the AMP iterates that follows from standard results in the AMP literature \cite{bayati2011dynamics}, \cite[Section 6.7]{feng2022unifying}. For any  Lipschitz test function $\varphi: \reals^{d} \times \reals^{d} \to \reals$ and $t\ge 1$, we almost surely have:
\begin{align}
    \label{eq:Lip_conv}
    \lim_{L\to \infty} \frac{1}{L}\sum_{\ell=1}^L
    \varphi(\bx_\ell, \bs_\ell^t)  = \E\left[\varphi (\bar{\bx}, \bar{\bx} + \bg^t)\right].
\end{align}
The claims in \eqref{eq:UER_thm} and \eqref{eq:BER_thm} require a test function $\varphi$ that is defined via indicator functions, which is not  Lipschitz. We handle this by sandwiching 
it between two Lipschitz functions that both converge to the required function in a suitable limit.

We prove \eqref{eq:BER_thm}, and omit  the proof of \eqref{eq:UER_thm} as it is simpler and follows along the same lines.
For $x_{\ell,j}\in \{\pm\sqrt{E}\}$, partition the space $\reals^d$ into two decision regions:
\begin{align}
    \mc{D}\left(x_{\ell,j}\right) :=\left\lbrace \bs_\ell^t: \left[h_t(\bs_\ell^t)\right]_j = x_{\ell,j}\right\rbrace,
\end{align}
and note that $ \ind\{\hat{x}_{\ell,j}^{t+1} = x_{\ell,j}\} = \ind\{\bs_\ell^t \in \mc{D}\left(x_{\ell,j}\right)\}$. Let $d\left(\bv, \mc{D}\right): = \inf \{\|\bv-\bu\|_2 : \bu\in\mc{D}\}$ denote the distance between a vector $\bv\in\reals^d$ and a set $\mc{D} \subset\reals^d.$ For any $\epsilon >0$, define  functions $\xi_{\epsilon,+}, \xi_{\epsilon,-}: \reals\times \reals^d \to \reals $ as follows:
\begin{align}
    &\xi_{\epsilon, +}\left(x_{\ell,j}, \bs_\ell^t\right) :=\begin{cases}
               1 , \quad\quad\quad\quad\quad\quad\quad \quad\text{if }\bs_\ell^t\in \mc{D}\left(x_{\ell,j}\right)\\
               0,  \quad\quad\quad\quad\quad\quad\quad \quad\text{if } d\big[\bs_\ell^t, \mc{D}\left(x_{\ell,j}\right)\big]> \epsilon\\
               1- d\big[\bs_\ell^t, \mc{D}\left(x_{\ell,j}\right)\big]/\epsilon, \;\text{otherwise}
            \end{cases} \nonumber\\
    &\xi_{\epsilon, -}\left(x_{\ell,j}, \bs_\ell^t\right) :=\begin{cases}
               1 , \quad\quad\quad\quad\quad\quad\qquad\; \text{if }d\big[\bs_\ell^t, \mc{D}\left(x_{\ell,j}\right)^c\big]> \epsilon\\
               0, \quad\quad\quad\quad\quad\quad\qquad\;\text{if } 
               \bs_\ell^t\in \mc{D}\left(x_{\ell,j}\right)^c\\
                d\big[\bs_\ell^t, \mc{D}\left(x_{\ell,j}\right)^c\big]/\epsilon, \, \ \quad\text{otherwise}
            \end{cases}       \nonumber.
\end{align}
Note that $\xi_{\epsilon,+}, \xi_{\epsilon,-}$ are Lipschitz-continuous with Lipschitz constant $1/\epsilon$. Moreover, define $\varphi_{\epsilon,+}: \reals^d\times \reals^d\to \reals$  as
\begin{align}
    \varphi_{\epsilon,+}(\bx_\ell, \bs_\ell^t):=\frac{1}{d}\sum_{j=1}^d \xi_{\epsilon,+}\left( x_{\ell,j}, \bs_\ell^t\right),
\end{align}
and define $\varphi_{\epsilon,-}$ analogously. Then  $\varphi_{\epsilon,+}$ and $  \varphi_{\epsilon,-}$, being sums   Lipschitz of functions, are also Lipschitz. 

For any $\epsilon>0$ and $\ell \in [L]$, we have
\begin{align}
 \varphi_{\epsilon,-}(\bx_\ell, \bs_\ell^t) \le  \frac{1}{d}\sum_{j=1}^d \ind\left\lbrace \bs_\ell^t \in \mc{D}\left(x_{\ell,j}\right)\right\rbrace \le \varphi_{\epsilon,+}(\bx_\ell, \bs_\ell^t).\label{eq:bound_rowwise_BER}
\end{align}
Applying \eqref{eq:Lip_conv} to $\varphi_{\epsilon,-}$ and $\varphi_{\epsilon,-}$, we have
\begin{equation}
\label{eq:psi_eps}
   \begin{split}
  & \lim_{L \to \infty} \frac{1}{L}\sum_{\ell=1}^L  \varphi_{\epsilon,-} (\bx_\ell, \bs_\ell^t) = \E\left[\varphi_{\epsilon,-}(\bar{\bx}, \bar{\bx}+\bg^t)\right], \\
     & \lim_{L \to \infty}  \frac{1}{L}\sum_{\ell=1}^L  \varphi_{\epsilon,+} (\bx_\ell, \bs_\ell^t) = \E\left[\varphi_{\epsilon,+}(\bar{\bx}, \bar{\bx}+\bg^t)\right].
   \end{split} 
\end{equation}
The functions $\varphi_{\epsilon,-}(\bx_\ell, \bs_\ell^t)$ and $\varphi_{\epsilon,+}(\bx_\ell, \bs_\ell^t)$ both converge pointwise to  $ \frac{1}{d}\sum_{j=1}^d \ind\left\lbrace \bs_\ell^t \in \mc{D}\left(x_{\ell,j}\right)\right\rbrace $ as $\epsilon\to0.$ Thus by the Dominated Convergence Theorem, we have 
\begin{align}
    \lim_{\epsilon\to0}
\E\left[\varphi_{\epsilon,-}(\bar{\bx}, \bar{\bx}+\bg^t)\right] 
 = \frac{1}{d}\sum_{j=1}^d\E\left[\ind\left\lbrace \bar{\bx}+\bg^t\in \mc{D}(\bar{x}_j)\right\rbrace\right] 
= \lim_{\epsilon\to0}
\E\left[\varphi_{\epsilon,+}(\bar{\bx}, \bar{\bx}+\bg^t)\right].\label{eq:DCT}
\end{align}

From \eqref{eq:bound_rowwise_BER}, \eqref{eq:psi_eps}, and \eqref{eq:DCT}, we conclude  that almost surely:
\begin{align}
 & \lim_{\epsilon \to 0} \lim_{L \to \infty} \frac{1}{L}\sum_{\ell=1}^L  \varphi_{\epsilon,-} (\bx_\ell, \bs_\ell^t)
  =  \lim_{\epsilon \to 0} \lim_{L \to \infty} \frac{1}{L}\sum_{\ell=1}^L  \varphi_{\epsilon,+} (\bx_\ell, \bs_\ell^t) \nonumber \\
 &  =  \lim_{L\to \infty} \frac{1}{Ld}\sum_{\ell=1}^L\sum_{j=1}^d \ind\left\lbrace \bs_\ell^t \in \mc{D}\left(x_{\ell,j}\right)\right\rbrace \nonumber\\
   & =\frac{1}{d}\sum_{j=1}^d\E\left[\ind\left\lbrace \bar{\bx}+\bg^t\in \mc{D}(\bar{x}_j)\right\rbrace\right]. \label{eq:proof_BER}
\end{align}
By recalling $ \ind\{\bs_\ell^t \in \mc{D}(x_{\ell,j})\}=\ind\{\hat{x}_{\ell,j}^{t+1} = x_{\ell,j}\} $ and noticing $\ind\{\bar{\bx}+\bg^t\in \mc{D}(\bar{x}_j)\} = \ind\{[h_t(\bar{x}+g^t)]_j = \bar{x}_j\}$, we see that  
\eqref{eq:proof_BER} is equivalent to \eqref{eq:BER_thm}. \qed

\subsection{Proof of Theorem \ref{thm:UER_BER_sc}}
\label{sec:proof_sc}
A state evolution result for the SC-AMP iteration applied to the generic spatially coupled linear model $\bY = \bA \bX + \bEps$ was obtained in \cite{liu2024RA}. In this model, $\bA \in \reals^{\tn\times L}$ is a spatially coupled Gaussian design matrix,  and $\bX \in \reals^{L \times d}$ and $\bEps \in \reals^{\tn \times d} $ are matrices whose row-wise empirical distributions converge in Wasserstein distance to well-defined limits. 

These assumptions are satisfied in our setting since $\bX$ and $\bEps$ are both row-wise i.i.d.
Applying the state evolution result in \cite{liu2024RA},   for any  Lipschitz test function $\varphi: \reals^{d} \times \reals^{d} \to \reals$,  $t\ge 1$ and $\sfc\in[\C]$, we have:
    \begin{align}\label{eq:Lip_conv_sc}
    \lim_{L\to \infty} \frac{1}{L/\C}\sum_{\ell\in\mc{L}_{\sfc}}
    \varphi(\bx_\ell, \bs_\ell^t)  = \E\left[\varphi (\bar{\bx}, \bar{\bx} + \bg_\sfc^t)\right],
    \end{align}
    where $\mc{L}_\sfc= \{(\sfc-1)L/\C+1, \dots, \sfc L/\C\}$. 
    
    Finally, as in the proof of Theorem \ref{thm:UER_BER}, we apply \eqref{eq:Lip_conv_sc} to Lipschitz approximations of suitable indicator functions and combine with a sandwich argument   to  obtain \eqref{eq:UER_sc_thm}--\eqref{eq:BER_sc_thm}. \qed

\section{Discussion and future directions}
This paper investigated   communication  over the GMAC in the  many-user  regime, where the number of users $L$ scales linearly with the codelength $n$. We  proposed a CDMA-type concatenated coding scheme with an efficient AMP decoder that can be tailored to the outer code. The asymptotic performance tradeoff was rigorously characterized, and it was shown that the scheme achieves state-of-the-art error performance  with an outer LDPC code and AMP decoding with a belief propagation (BP) denoiser. 

Although we considered Gaussian design matrices $\bA$ (i.i.d.\@ or spatially coupled), using recent  results on AMP universality \cite{Wan22, Tan23d}, the decoding algorithm and all the theoretical results remain valid for a much broader class of `generalized white noise' matrices. This class includes i.i.d.\@ sub-Gaussian matrices, so the results apply to the popular setting of random binary-valued signature sequences \cite{shamaiVerdu01, guo2005randomly}. 

The concatenated scheme and its analysis can also be extended to nonbinary CDMA constellations. Based on results for sparse superposition coding \cite{hsieh2021modulated,hsieh2022near},   we expect that using PSK constellations instead of binary will improve the performance tradeoff for complex-valued channels.

Our theoretical results for the AMP decoder with a BP denoiser  are valid only when  BP  is executed for fewer rounds than the girth of the LDPC factor graph (Lemma \ref{lem:etaBP_deriv}). However,  simulations indicate that the performance of AMP decoding may be  improved if BP is run for additional rounds in each  denoising step. Precisely characterizing the performance of AMP when the number of BP rounds is larger than the girth of the factor graph is a challenging open question. Another open question is to characterize the performance improvement achieved by executing several rounds of BP after AMP has converged. 

Our asymptotic performance curves were obtained by running the state evolution recursion until convergence, with each iteration involving a  computationally intensive MSE estimation step (see \eqref{eq:SE_Sigma_k+1} and \eqref{eq:SC_SE_Psi_t}). An alternative approach would be to directly characterize the fixed points of state evolution for sufficiently large base matrices using a potential function analysis, similar to  \cite{hsieh2022near, kowshik2022improved}. The main challenge  is that in our setting, such an analysis would require finding the extremal values of a potential function over (positive-definite) $d  \times d$ matrices \cite{aubin2018committee}, which is infeasible even for moderately large $d$. One approach to simplify this optimization would be to assume symmetries in the structure of the optimal matrix.

This paper focused on the conventional GMAC setting where users have distinct codebooks. It would be very interesting to adapt the coding scheme  to  \emph{unsourced}   random access  \cite{polyanskiy2017perspective, Amalladinne2020Coded,fengler2021SPARCS,amalladinne2022unsourced}, where the users share a common codebook and only a subset of users are active. 
    
 \paragraph*{Acknowledgment}
% %
X.\@ Liu was supported in part by a Schlumberger Cambridge International Scholarship funded by the Cambridge Trust.
 We thank Jossy Sayir for sharing his implementation of the belief propagation LDPC decoder, and Pablo Pascual Cobo for helping with the implementation of the SC-AMP decoder. We thank Jean-Francois Chamberland for noticing an error in an earlier version of Lemma \ref{lem:etaBP_deriv}.

\bibliographystyle{IEEEtran}
{\small \bibliography{coding} } 
\end{document}